\documentclass[10pt]{article}
\usepackage[margin=.8in,left=.8in]{geometry}

\usepackage{amsthm}
\usepackage{amssymb}
\usepackage{amsmath}
\usepackage{mathtools}
\usepackage{adjustbox}

\usepackage[table]{xcolor}

\usepackage{stmaryrd}    
\usepackage{xfrac}
\usepackage{enumitem}\setlist{nosep}

\usepackage[safe]{tipa} 

\usepackage{hhline}

\usepackage{tensor}

\usepackage{tikz}
\usetikzlibrary{
  backgrounds, 
  decorations.pathmorphing, 
  decorations.markings,
  decorations.text
}
\usetikzlibrary{cd}

\usepackage{mathrsfs} 
\DeclareMathAlphabet{\mathpzc}{OT1}{pzc}{m}{it} 

\usepackage{hyperref}
\hypersetup{
  colorlinks = true,
  linkcolor= darkblue, citecolor=darkblue            
}

\hypersetup{
        colorlinks=true,
        linkcolor=darkblue,
        filecolor=darkblue,      
        urlcolor=black,
        citecolor=darkblue,
        linktoc=page
    }

\usepackage{cleveref}

\definecolor{darkblue}{rgb}{0.05,0.25,0.65}
\definecolor{darkgreen}{RGB}{20,140,10}
\definecolor{lightgray}{rgb}{0.9,0.9,0.9}
\definecolor{darkorange}{RGB}{200,100,5}
\definecolor{darkyellow}{rgb}{.91,.91,0}
\definecolor{orangeii}{RGB}{200,100,5}
\definecolor{lightblue}{RGB}{243, 250, 255}

\newtheorem{theorem}{Theorem}[section]

\newtheorem{lemma}[theorem]{Lemma}

\newtheorem{proposition}[theorem]{Proposition}
\newtheorem{corollary}[theorem]{Corollary}
\theoremstyle{definition}
\newtheorem{definition}[theorem]{Definition}

\newtheorem{example}[theorem]{Example}

\newtheorem{remark}[theorem]{Remark}

\usepackage{accents}
\newlength{\dhatheight}

\let\PLAINthebibliography\thebibliography
\renewcommand\thebibliography[1]{
  \PLAINthebibliography{#1}
  \setlength{\parskip}{0.5pt}
  \setlength{\itemsep}{0.5pt plus .3ex}
}

\newcommand{\differential}{\mathrm{d}}

\newcommand\bos[1]{\mathstrut\mkern2.5mu#1\mkern-14mu\raise1.7ex%
  \hbox{$\scriptstyle\rightsquigarrow$}}

\newcommand\bosonic[1]{\mathstrut\mkern2.5mu#1\mkern-14mu\raise1.7ex%
  \hbox{$\scriptstyle\rightsquigarrow$}}

\usepackage[cmtip,all]{xy}
\newcommand{\longsquiggly}{\xymatrix{{}\ar@{~>}[r]&{}}}

\usepackage{bbm} 

\usepackage[new]{old-arrows}   

\newcommand{\FR}{\mathbb{R}}

\newcommand{\dd}{\mathrm{d}}

\newcommand{\even}{\mathrm{even}}

\newcommand{\odd}{\mathrm{odd}}

\newcommand{\om}{\omega}

\newcommand{\wtom}
{\widetilde{\omega}}

\newcommand{\wtxi}{\widetilde{\xi}}

\newcommand{\CX}{\mathcal{X}}

\newcommand{\frg}{\mathfrak{g}}

\begin{document}

\setlength{\abovedisplayskip}{3pt}
\setlength{\belowdisplayskip}{3pt}
\setlength{\abovedisplayshortskip}{-3pt}
\setlength{\belowdisplayshortskip}{3pt}

\title{
Covariant Lie Derivatives and (Super-)Gravity}

\author{
  Grigorios Giotopoulos${}^{\ast}$
  \;\;
  \;\;
}

\maketitle

\thispagestyle{empty}

\begin{abstract}
The slightly subtle notion of covariant Lie derivatives of \textit{bundle-valued} differential forms is crucial in many applications in physics, notably in the computation of conserved currents in gauge theories, and yet the literature on the topic has remained fragmentary. This note provides a complete and concise mathematical account  of covariant Lie derivatives on a spacetime (super-)manifold $M,$ defined via choices of lifts of spacetime vector fields to principal $G$-bundles over it, or equivalently, choices of covariantization correction terms on spacetime. As an application in the context of (super-)gravity, two important examples of covariant Lie derivatives are presented in detail, which have not appeared in unison and direct comparison: $\bf{(i)}$ The natural covariant Lie derivative relating (super-)diffeomorphism invariance to local translational (super-)symmetry, and $\bf{(ii)}$ the Kosmann Lie derivative relevant to the description of isometries of (super-)gravity backgrounds. Finally, we use the latter to rigorously justify the usage of the traditional (non-covariant) Lie derivative on coframes and associated fields in dimensional reduction scenarios along abelian $G$-fibers, an issue which has thus far remained open for topologically non-trivial spacetimes.
\end{abstract}

\vspace{.1cm}

\begin{center}
\begin{minipage}{13cm}
\tableofcontents
\end{minipage}
\end{center}

\medskip

\vfill

\hrule
\vspace{5pt}

{
\footnotesize
\noindent
\def\arraystretch{1}
\tabcolsep=0pt
\begin{tabular}{ll}
${}^*$\,
&
Mathematics, Division of Science; and
\\
&
Center for Quantum and Topological Systems,
\\
&
NYUAD Research Institute,
\\
&
New York University Abu Dhabi, UAE.  
\end{tabular}
\hfill
\href{https://ncatlab.org/nlab/show/Center+for+Quantum+and+Topological+Systems}{
\adjustbox{raise=-15pt}{
\includegraphics[width=3cm]{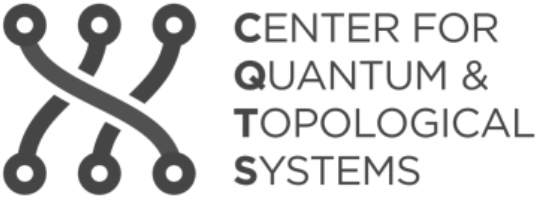}
}
}

\vspace{1mm} 
\noindent 

\vspace{.2cm}

\noindent
The author acknowledges the support by {\it Tamkeen} under the 
{\it NYU Abu Dhabi Research Institute grant} {\tt CG008}.
}

\newpage

\section{Introduction \& Overview}

On spacetime manifolds $X$ with non-trivial topologies, ``matter fields''\footnote{By which we mean all fields transforming in  non-trivial representations of the corresponding gauge groups, such as the Dirac fermion field $\psi_{\mathrm{Dir}}$ of electrodynamics or the coframe $e$ and the gravitino $\psi$ of supergravity (thus excluding the respective gauge fields).} $\phi$ are globally seen as differential forms valued in potentially non-trivial vector bundles (see e.g. \cite{Tu17} for relevant background on differential geometry)
\begin{align}\label{AssociatedVectorBundle}
V_G \, = \,  P\times_G V\, ,
\end{align}
associated to a principal $G$-bundle $\pi : P \rightarrow X$ via some representation of $G$ on a vector space $V$, 
$$\phi \quad \in \quad \Omega^p(X; \, V_G)\, .$$
On a local trivializing cover $\{U_i\hookrightarrow X\}_{i\in I}$ of the bundles, these are represented  by families of locally defined $p$-forms $\{\phi_i \, \in \,  \Omega^p(U_i; \,V) \}_{i\in I}$ valued in the vector space $V$, related on any overlap $U_{ij}=U_i\cap U_j$ by the action of the corresponding transition function $\gamma_{ij} \in \Omega^0(U_{ij}; \, G ) $
$$
\phi_j = \gamma_{ji} \cdot \phi_i \, ,
$$
via the given $G$-representation on $V$. 

Given any vector field $\xi \in \mathcal{X}(X)$, the ordinary Lie derivative operation $L_\xi$ does not yield a well-defined infinitesimal transformation $$ \delta_\xi \phi \quad \in \quad  T_\phi \big(\Omega^p(X; \, V_G)\big)\cong \Omega^p(X; \, V_G) $$ of the fields\footnote{In geometrical terms, it does not define a vector field on the field space $\Omega^p(X; \, V_G)= \Gamma_X(\Lambda^pT^*X \otimes V_G)$ \cite{GS25}.}. That is, the locally defined expression 
$$
L_\xi \phi_i \,  := \, \differential \iota_{\xi} \phi_i + \iota_{\xi} \differential \phi_i $$
is easily seen to be, in general, \textit{non-covariant} under the action of the transition functions
\begin{align}\label{NonCovarianceLieDerivativeMatterField}
L_\xi \phi_j =  L_\xi (\gamma_{ji} \cdot \phi_i) \quad \neq \quad  \gamma_{ji} \cdot L_\xi \phi_i\, ,
\end{align}
and hence the family of locally defined objects $\{L_\xi \phi_i\}_{i\in I}$ does not patch up to a globally defined section of $\Omega^p(X; \, V_G)$, to be interpreted as a tangent vector at $\phi \in  \Omega^p(X; \, V_G)$. 

Similarly,  for any globally defined connection 1-form $\widetilde{\omega} \in \Omega^1(P;\,\mathfrak{g})$ on the total space $P$, such as the spin connection of (super-)gravity, and a family of its local ``gauge field'' representatives $\{\omega_i \, \in \, \Omega^1( U_i;\, \mathfrak{g}) \}_{i\in I}$ on the trivializing neighborhoods  $\{U_i \hookrightarrow X\}_{i\in I}$, it is equally easily seen that, in general, 
\begin{align}\label{NonCovarianceLieDerivativeConnection}
L_\xi (\omega_j) = L_\xi (\gamma_{ji} \cdot \omega_i \cdot \gamma_{ij} + \gamma_{ji} \cdot \differential \gamma_{ij} ) \quad \neq \quad   \gamma_{ji} \cdot L_\xi \omega_i \cdot \gamma_{ij}\, . 
\end{align}
Consequently, the family of locally defined objects $\{L_\xi \omega_i\}_{i\in I}$ does not define a section of $\Omega^1(M;\,P\times_G \mathfrak{g})$, to be interpreted as a tangent vector in  $T_\omega \mathrm{Conn}_G(P)$.

\medskip

\noindent {\bf Issues related to the non-covariance of the Lie derivative}

Despite the above glaring inconsistencies, the majority of literature on theoretical physics tends to ignore this issue. From the mathematical perspective, this can only be justified in severely restrictive scenarios, such as if the theories are considered only in their trivial topological sector. For instance, this is the case (necessarily so) if these theories are defined over a contractible spacetime, such as flat Minkowski spacetime $X\cong \mathbb{R}^{1,d}$. Our overarching motivation for these considerations stems from flux quantization in supergravity \cite{GSS24-SuGra}\cite{GSS-M5Brane}, a construction which yields non-trivial new ``gluing'' information only over non-contractible spacetimes \cite{SS24Flux}\cite{SS23FQ}. As such, the traditional Lie derivative is in general (a priori) inappropriate and requires further justification in each relevant context. A critical example is dimensional reduction in the coframe formalism along abelian $G$-symmetries (cf. Lem. \ref{InvarianceViaKosmannAndNaiveLieDerivativeOfCoframe}, Cor. \ref{InvarianceOfAssociatedFields}), which is essential for our upcoming application in the dimensional reduction of 11D $S^4$-flux quantized super-space supergravity \cite{GSS24-SuGra} to 10D $\mathrm{cyc}(S^4)$-flux quantized IIA supergravity \cite{GS25b}.

Moreover, even on trivial topologies, the usual Lie derivative approach is burdened with further issues from the perspective of dynamics in traditional gauge field theories. This is because the non-covariance with internal gauge transformations (Eq. \eqref{NonCovarianceLieDerivativeMatterField}\eqref{NonCovarianceLieDerivativeConnection}) implies that conserved Noether currents of symmetries generated by spacetime vector fields are \textit{not gauge-invariant}, when computed via the traditional Lie derivative as an infinitesimal transformation of the fields (see e.g. \cite{Ja80} in the context of electromagnetism, and \cite[Eq. 6.8]{OG06} in the context of gravity).

\medskip

\noindent{\bf Resolutions via covariantization corrections and lifts to associated principal G-bundles}

We now survey approaches in the literature that deal with this issue, albeit mostly with a motivation originating from the dynamical considerations of gauge theoretic fields, and highlight their commonalities and discrepancies.

In the context of spacetime symmetries of $G$-gauge fields on fixed gravitational backgrounds, \cite{Ja80} advocates a natural gauge-covariant modification of the Lie derivative of a gauge field $\om$, along any vector field $\xi \in \CX(X)$, by subtracting a specially chosen gauge transformation with parameter $\iota_\xi \omega$
\begin{align*}
L^{\om}_\xi \om \, :=  \,   L_\xi \omega - \dd^\omega(\iota_\xi \omega)   \,=\, \iota_\xi F_\omega  \, ,
\end{align*}
where $F_\omega = \dd \om + \frac{1}{2}[\om,\om]$ is the corresponding curvature 2-form.
The action is extended to associated matter fields via
$$
L_\xi^\om \phi := L_\xi \phi + \iota_\xi \om \cdot \phi \, .
$$
Under Noether's first theorem,   this transformation directly yields the correct gauge-invariant conserved currents, bypassing the otherwise ad-hoc ``Belinfante procedure'' \cite{Be40}.
In \cite{Ja80} these transformations are termed (infinitesimal) ``\textit{gauge-covariant coordinate transformations}''. They often go under the same name in the context of supergravity (see e.g. \cite{VF12}\cite{Or15}), where they are motivated by their appearance in the closure of local supersymmetry transformations.\footnote{Note, rather than the global well-definiteness on non-trivial topologies, the motivation in the above sources is the dynamical consideration of gauge fields.} The relation of this covariant Lie derivative and local supersymmetry is in fact deeper (cf. Sec \ref{SuperDiffeomorphismSymmetrySec}). In \cite{OG06} and \cite{CGRS20} this transformation is termed the \textit{``(natural) covariant Lie derivative''} associated to the gauge field $\om$, a nomenclature we adopt in the current text for its transparent geometrical origin (cf. Ex. \ref{LiftingViaChosenConnection}).

With the purpose of clarifying conservation laws in gravitational theories, \cite{OG06} notices that there are in fact plenty of mathematically acceptable gauge-covariant modifications of the traditional Lie derivative, determined by 0-forms $B_\xi$  valued in the Lie algebra $\mathfrak{g}$, linear in $\xi \in \CX(X)$, which `transform as a gauge field'. Interpreted more precisely, these correspond to families of locally defined $0$-forms $
\big\{(B_{\xi})_i \quad \in \quad  \Omega^0(U_i; \, \mathfrak{g})\big\}_{i\in I}
$ over a trivializing cover $\{U_i \hookrightarrow X\}_{i\in I}$ (cf. Eq. \eqref{CovariantizationTermSum}) related on overlaps as 
$$ (B_{\xi})_j = 
\gamma_{ji} \cdot (B_{\xi})_i \cdot \gamma_{ij} + \gamma_{ji} \cdot\iota_\xi \differential \gamma_{ij} \, . 
$$
The  corresponding covariant Lie derivatives are then given by
$$\widetilde{L}_\xi \phi = L_\xi \phi + B_\xi \cdot \phi \quad \mathrm{and} \quad \widetilde{L}_\xi \omega = L_\xi \omega - \dd^\omega B_\xi \, , $$
respectively.

In the context of first-order Einstein--Cartan gravity, among the myriads of choices \cite{OG06} finds that the ``\textit{Kosmann Lie derivative}'' (therein called the Yano derivative), defined via
$$
(B^K_{\xi})^a{}_{b} \, = \,   - \eta^{ad}\, (L_\xi e)_{[bd]} \quad \in \quad \Omega^{0}(U_i, \, \mathfrak{so}(1,d)\big)
$$
where $ e\quad \in \quad \Omega^{1}(X; \, \FR^{1,d}_{\mathrm{SO}(1,d)})$ is the coframe field (a.k.a. vielbein), a (non-degenerate) $1$-form valued in the vector bundle $\FR^{1,d}_{\mathrm{SO}(1,d)}\cong P \times_{ \mathrm{SO}(1,d)} \FR^{1,d}$ associated to a background $\mathrm{SO}(1,d)$-principal bundle \eqref{AssociatedVectorBundle} via the fundamental representation, recovers in particular the Komar mass formula as a conservation law. The same is noted in \cite{FF09}, where also a geometrical interpretation particular to $\mathrm{SO}(1,d)$-structures and the Kosmann Lie derivative is reviewed (cf. Rem. \ref{AnotherGeometricCharacterization}).
In \cite{JM15} this same covariant Lie derivative is termed the Lorentz-Lie derivative and is employed to recover the black hole entropy as a Noether charge.  Further considerations of black hole mechanics via the Kosmann Lie derivative are pursued in \cite{EMO20}, following \cite{Pr17}, and more recently in super-gravitational contexts in \cite{BO23}\cite{BMO24}\cite{BMO25}. In fact, this covariant Lie derivative has appeared  earlier in \cite{Or02} (under the name Lorentz-Lie derivative), where it is motivated as arising via the commutator of two \textit{Killing spinor} super-symmetry transformations for super-symmetic backgrounds. 

Crucially, the Kosmann Lie derivative satisfies 
$$
L_\xi g = 0 \quad \iff \quad L_\xi^K e = 0 \, ,
$$
where $g= \langle e, e\rangle = \eta_{ab} \, e^a\otimes e^b$ is the corresponding metric tensor and\footnote{Our convention on indices in the coframe basis is given by $L_\xi e^a = (L_\xi e)_b{}^a \, e^b$. Lowering (and raising) indices is via the Minkowski matrix elements $\eta_{ab}$ (and $\eta^{ab}$) of signature $(1,d)$, e.g., $(L_\xi e)_{bd} = (L_\xi e)_{b}{}^a \, \eta_{ad}$. Brackets and square brackets on indices denotes  symmetrization and antisymmetrization,  respectively, e.g., $(L_\xi e)_{[bd]}=\frac{1}{2}\big( (L_\xi e)_{bd}  - (L_\xi e)_{db} \big) $.  }
$$
L_\xi^K e^a  \, = \,   L_\xi e^a -  \eta^{ad} \, (L_\xi e)_{[bd]} \,  e^b \, .
$$ 
This property ensures that the corresponding conservation laws of the metric (2nd-order) and coframe (1st-order) formulations of gravity agree. Indeed, we shall adopt a slightly strengthened version of this as the \textit{defining property} of the Kosmann Lie derivative (Def. \ref{KosmannLift}) and recover the form of the covariantization-correction terms via Prop. \ref{TheKosmannLiftExists}. We follow the nomenclature of the ``Kosmann'' Lie derivative in honor of \cite{Ko72} who originally introduced it as a definition of a Lie derivative on spinors, albeit in an ad-hoc manner, following Lichnerowicz \cite{Li63} for the case of Killing vector fields. We briefly explain geometrically how the latter spinorial Lie derivative is recovered in Rem. \ref{LiftingToSpinCovers}.

In \cite{Pr17} with a motivation of studying gauge field and gravitational configurations with non-trivial topologies in the context of black hole mechanics, and in contrast to the spacetime description of \cite{JM15}, the total space principal G-bundle avatars of the fields are employed (cf. Eqs. \eqref{VectorBundleValuedAreFormsTotalSpaceHorizontalForms} \eqref{LocalGaugeFieldsAreFormsTotalSpaceVerticalForms}). Therein, the fields and their Lagrangian dynamics are all lifted to the corresponding principal $G$-bundle $P$ over the spacetime $M$, with the action of any spacetime vector field $\xi$ computed via the Lie derivative solely on the total space along \textit{arbitrary} vector fields on $P$ (cf. Eqs \eqref{TotalSpaceLieDerivativeMatterField}\eqref{TotalSpaceLieDerivativeConection}) that project to the given vector field $\xi$.  The relation to the spacetime covariant Lie derivatives of \cite{OG06} is left open, except for the sole case of the Kosmann Lie derivative under the (unnecessary) condition of the vector fields being infinitesimal isometries of the corresponding metric.

The purpose of this text is thus threefold: {\bf (i)} Providing a clear mathematical description encompassing all the above descriptions and their relation (Sec. \ref{CovariantizedActionOfSpacetimeVectorFields}); {\bf (ii)} Expanding on two different examples of covariant Lie derivatives of importance, both appearing in the context of a single theory, that of (super-)gravity (Sec. \ref{SuperDiffeomorphismSymmetrySec}, \ref{TheKosmannLieDerivativeAndIsometriesOfSuperSpacetimeBackgrounds}); {\bf(iii)} And lastly, justifying the usage of the (non-covariant) Lie derivative in the context of dimensional reduction along abelian $G$-fibers in the coframe formalism of (super-)gravitational backgrounds (Sec. \ref{IsometryInTheCoframeFormalismAndKaluzaKleinDimensionalReduction}).  


\medskip




\section{Covariantized action of spacetime vector fields}\label{CovariantizedActionOfSpacetimeVectorFields}
We now describe how to properly define infinitesimal transformations $\delta_{\xi} \phi, \, \delta_{\xi} \omega$ for matter fields $\phi$ and gauge fields $\om$ along any spacetime vector field $\xi \in \mathcal{X}(X)$. We expand on the geometric foundation of this construction via the total space of the underlying principal $G$-bundle $P$, while also pinpointing the space of choices involved, and at the same time making explicit the relation to the local formulas appearing in the literature when computing such ``covariant Lie derivatives'' directly on the base spacetime. The basic concepts and facts we appeal to are completely standard in principal $G$-bundle theory within the category of smooth manifolds (see e.g. \cite[Ch. 6]{Tu17}). In fact, they are formally exactly the same within the category of super manifolds\footnote{Modulo certain technical details which we shall not enter into here. For instance, all field theoretic statements here should really involve $\Sigma$-parametrized families of sections of super-bundles and connections 1-forms, where $\Sigma$ is an arbitrary ``probe'' supermanifold. In other words, in the context of fermionic fields one must implicitly consider $\Sigma$-plots of fields as per \cite{Schreiber24}\cite{Gi25}\cite{GSS24-SuGra}\cite{GSS26}(following \cite{dcct}, in turn in the vein of \cite{Grothendieck73}).} (see e.g. \cite{Eder21} for a rigorous account), as  necessary towards the applications to super-gravity via its super-spacetime formulation in Sec. \ref{CovariantLieDerivativesOfSupergravitySec}. Thus, for the sake of brevity, we shall refrain from using the adjective ``super'' throughout this section. 

The crux of the construction is the canonical bijection of $p$-forms on the base $X$ valued in an associated vector bundle $V_G$, and horizontal $G$-equivariant $p$-forms on the total space $P$ valued in the vector space $V$\footnote{A $p$-form $\omega$ is horizontal if $\iota_Z\omega=0$ for all vertical vector fields $Z\in \Gamma_P(VP)\hookrightarrow \CX(P)$. It is $G$-equivariant if $\rho_{g}^*\om = g^{-1} \triangleright \omega$ where $\rho:P\times G\rightarrow P$ is the right action of $G$ on $P$ and $\triangleright: G\times V\rightarrow V$ is the left (linear) action on $V$.}
\begin{align}\label{VectorBundleValuedAreFormsTotalSpaceHorizontalForms}
\big\{ \phi \, \in \, \Omega^{p}(X;\, V_G)\big\} \quad \cong \quad \big\{ \widetilde{\phi} \, \in \, \Omega^{p}(P;\, V) 
\big\}_{\mathrm{Hor}, G\mbox{-}\mathrm{equiv}} \, .
\end{align}
We recall that in terms of the families of  local representatives $\{\phi_i \, \in \, \Omega^p(U_i; \, V)\, \vert \, \phi_j = \gamma_{ji}\cdot  \phi_i\}_{i,j\in I}$ of each $\phi$ over local trivializing covers of the bundles, this bijection is realized by pulling back to the base via the canonical local section $\sigma_i: U_i \rightarrow P$ associated to any trivialization $P|_{U_i}\cong U_i \times G$ 
$$
\{\phi_i \}_{i\in I}\equiv \{\sigma_i^* \widetilde{\phi} \}_{i\in I}\quad  \longmapsfrom \quad \widetilde{\phi} \, .
$$

On the other hand, families of representative gauge fields on $X$ over local trivializing covers are in canonical bijection with  connections on $P$, i.e., $G$-equivariant forms valued in the Lie algebra $\mathfrak{g}$
\begin{align}\label{LocalGaugeFieldsAreFormsTotalSpaceVerticalForms}
\big\{ \{\omega_i \, \in \, \Omega^1( U_i; \, \mathfrak{g}) \, \vert \, \omega_j= \gamma_{ji} \cdot \omega_i \cdot \gamma_{ij} + \gamma_{ji} \cdot \differential \gamma_{ij}  \}_{i,j\in I}\big\} \quad \cong \quad \mathrm{Conn}_G(P)\longhookrightarrow 
\big\{\widetilde{\omega} \, \in \, \Omega^{1}(P;\, \mathfrak{g}) \big\}_{G\mbox{-}\mathrm{equiv}} \, .
\end{align}
which furthermore satisfy $\om(A^\#)= A$ for all fundamental vector fields $A^\# \in \CX(P)$ generated by any Lie algebra element $A\in \frg$. With respect to the associated splitting 
$$TP \quad \cong_{\widetilde{\omega}}\quad HP \oplus VP\, , $$
with $HP= \mathrm{ker}(\widetilde{\om})\cong \pi^*TM$, each connection form is in particular a \textit{vertical} 1-form. As with the case of horizontal and $G$-equivariant $p$-forms \eqref{VectorBundleValuedAreFormsTotalSpaceHorizontalForms}, this bijection is realized by pulling back via canonical local sections associated to local trivializations of the bundle
$$
\{\om_i \}\equiv \{\sigma_i^* \widetilde{\om} \}_{i\in I}\quad  \longmapsfrom \quad \widetilde{\om} \, .
$$
\medskip
\noindent{\bf The total space Lie derivative}

Under these canonical identifications for the fields the traditional notion of Lie derivatives makes sense, with the caveat that this is now taken \textit{along vector fields on $P$} where we are instead dealing with \textit{globally defined forms valued in plain vector spaces}. More explicitly, for any \textit{$G$-invariant vector field}\footnote{These constitute the Lie algebra corresponding to the automorphism group of $P$.} $Z\in \mathcal{X}(P)_G$ we have
\begin{align}\label{TotalSpaceLieDerivativeMatterField}
L_{Z} \widetilde{\phi}&=\dd \iota_{Z} \widetilde{\phi} + \iota_Z \dd \widetilde{\phi} \nonumber \\
&= \dd \iota_{Z} \widetilde{\phi} -\iota_{Z}( \widetilde{\omega}\wedge\widetilde{\phi}) +\iota_{Z} (\widetilde{\omega}\wedge\widetilde{\phi})  + \iota_Z \dd \widetilde{\phi}\\
&= \iota_Z \dd^{\widetilde{\omega}}\widetilde{\phi} + \dd^{\widetilde{\omega}}(\iota_{Z} \widetilde{\phi}) - \iota_Z \widetilde{\omega} \cdot \widetilde{\phi} \nonumber
\end{align}
and similarly 
\begin{align}\label{TotalSpaceLieDerivativeConection}
L_Z \widetilde{\omega} &= \dd \iota_{Z} \widetilde{\omega} + \iota_Z \dd \widetilde{\omega} \nonumber \\ 
&= \dd \iota_{Z} \widetilde{\omega} -\frac{1}{2} \iota_{Z} [\widetilde{\omega}, \widetilde{\omega}]+\frac{1}{2} \iota_{Z} [\widetilde{\omega}, \widetilde{\omega}] + \iota_Z \dd \widetilde{\omega} \\
&= \iota_Z F_{\widetilde{\omega}} + \dd^{\widetilde{\omega}}(\iota_{Z} \widetilde{\omega})
\,,\nonumber  
\end{align}
where $F_{\widetilde{\omega}}$ is the corresponding (horizontal) curvature 2-form and $\dd^{\widetilde{\omega}} (\iota_Z \widetilde{\omega})$, $\dd^{\widetilde{\omega}} \widetilde{\phi}$ are the covariant derivatives of the corresponding horizontal $G$-equivariant forms. Notice that each of the terms in the final expressions is manifestly horizontal and $G$-equivariant for each $G$-invariant vector field Z, i.e., 
\begin{align*}
L_Z \widetilde{\phi} \quad &\in \quad  \Omega^{1}(P; \, V)_{\mathrm{Hor}, G\mbox{-}\mathrm{equiv}}\\
L_Z \widetilde{\omega} \quad &\in \quad \Omega^{1}(P; \, \mathfrak{g})_{\mathrm{Hor}, G\mbox{-}\mathrm{equiv}}
 \, .
\end{align*}
It follows that under the canonical bijection of \eqref{VectorBundleValuedAreFormsTotalSpaceHorizontalForms}, pulling down via local sections, these yield infinitesimal transformations of the original spacetime fields
\begin{align*}
\delta_{Z} \phi \quad &\in \quad \Omega^1(X; \, V_G) \, \cong \, T_\phi \big(\Omega^1(X; \, V_G)\big) 
\\
\delta_Z \omega \quad &\in \quad \Omega^1(X, \, P\times_G\mathfrak{g}) \, \cong \,  T_\omega \big(\mathrm{Conn}_G(P)\big)
\, .
\end{align*}

Summarizing, the Lie derivative along $G$-invariant vector fields on the total space $P$ does indeed define infinitesimal transformations of gauge theoretic fields. The remaining task then is to associate a $G$-invariant lift $\widetilde{\xi}\in \CX(P)_G$ to any vector field on spacetime $\xi \in \mathcal{X}(X)$, in that
\begin{align}\label{LiftingCondition}
\pi_{*p} (\widetilde{\xi}_p) \, = \, \xi_{\pi(p)} \qquad \mathrm{and} \qquad (\rho_g)_{*p}(\widetilde{\xi}_p) = \widetilde{\xi}_{p\cdot g}, \, ,
\end{align}
for all $p\in P$, where $\pi : P\rightarrow X $ is the bundle projection and $\rho: P\times G\rightarrow P$ is the right $G$-action on the total space. The lift to total space vector fields should be $\FR$-\textit{linear}
$
\widetilde{\xi_1+\xi_2}\, = \,  \widetilde{\xi_1} + \widetilde{\xi_2}
$
so that moreover 
$$
\delta_{\widetilde{\xi_1+\xi_2}}\phi\, = \, \delta_{\widetilde{\xi}_1} \phi+ \delta_{\widetilde{\xi}_2} \phi\,, 
$$
as required for a notion of infinitesimal transformations of fields, parametrized by vector fields on the base. \footnote{Note that the latter property defines, in turn, a morphism of $\FR$-vector spaces $$
\delta_{\widetilde{(-)}} \, : \, \CX(X) \longrightarrow \CX(\mathcal{F})
$$
from vector fields on $X$ into vector fields on either field space $\mathcal{F} = \Omega^{1}(P, \, V)_{\mathrm{Hor}, G\mbox{-}\mathrm{equiv}}$ or $\mathcal{F}  = \mathrm{Conn}_G(P)\hookrightarrow \Omega^{1}(P, \, \mathfrak{g})_{G\mbox{-}\mathrm{equiv}}$. In local field theory such a map should have image in \textit{local vector fields} $\CX_\mathrm{loc}(\mathcal{F})\hookrightarrow \CX(\mathcal{F})$, which enjoy a Lie algebra structure (\cite[Def. 6.7]{GS25}). Nevertheless, the map need not be a morphism of Lie algebras.}
The lift $\widetilde{(-)}:\CX(X)\rightarrow \CX(P)_G$ constitutes precisely a \textit{choice} of covariantization for the action of spacetime vector fields.
\begin{definition}[\bf Covariant Lie derivative]\label{CovariantLieDerivative}
Let $P\rightarrow X$ be a principal $G$-bundle, and 
\begin{align*}
\widetilde{(-)}\, : \, \mathcal{X}(X) &\longrightarrow \mathcal{X}(P)_G \\
\xi &\longmapsto \widetilde{\xi}
\end{align*}
a \textit{chosen} lifting \eqref{LiftingCondition} $\mathbb{R}$-linear map of spacetime vector fields into $G$-invariant vector fields on $P$.
The \textit{covariant Lie derivatives associated to the lift $\widetilde{(-)}$} 
of any matter field $\phi$ and connection $\omega$ on X along a spacetime vector field $\xi \in \CX(X)$, are defined as the sections 
$$\widetilde{L}_{\xi} \phi \quad \in \quad \Omega^p(X; \, V_G)
\quad \quad \mathrm{and} \quad \quad \widetilde{L}_{\xi} \omega \quad \in \quad \Omega^1(X; \, P\times_G\mathfrak{g}) 
$$
corresponding, via \eqref{VectorBundleValuedAreFormsTotalSpaceHorizontalForms}, to the Lie derivatives on the total space
$$
L_{\widetilde{\xi}} \widetilde{\phi} \quad \in \quad  \Omega^{p}(P; \, V)_{\mathrm{Hor}, G\mbox{-}\mathrm{equiv}} \quad \quad \mathrm{and} \quad \quad
 L_{\widetilde{\xi}} \widetilde{\omega} \quad \in \quad \Omega^{1}(P; \, \mathfrak{g})_{\mathrm{Hor}, G\mbox{-}\mathrm{equiv}}\, ,
$$
respectively. In terms of local representatives, this means that 
\begin{align}\label{PullingBackTotalSpaceCovariantLieDerivative}
(\widetilde{L}_{\xi} \phi)_i \, := \, \sigma_i^* (L_{\widetilde{\xi}} \widetilde{\phi}) \quad \quad \mathrm{and} \quad \quad   (\widetilde{L}_{\xi} \omega)_i \, := \, \sigma_i^* (L_{\widetilde{\xi}} \widetilde{\omega}) \, ,
\end{align}
for any local section $\sigma_i : U_i \rightarrow P$ associated to a local trivialization of the bundle.
\end{definition}
Although the above definition is rigorous and fully general, it does not provide an explicit calculational method nor a relation to the traditional would-be Lie derivative on X. To that end:
\begin{lemma}[\bf Covariant Lie derivative on the base via a connection]\label{CovariantLieDerivativeOnThebAseViaAConnection}
Let $\widetilde{(-)}\, : \, \mathcal{X}(X) \longrightarrow \mathcal{X}(P)_G $ be an $\FR$-linear lift and $\wtom$ any connection on $P$ with associated splitting
$$TP \quad \cong_{\widetilde{\omega}}\quad HP \oplus VP\, , $$
where $HP:= \mathrm{ker}(\widetilde{\om})\cong \pi^*TM$ is the induced horizontal subbundle. Denote the induced  decomposition of any lifted $G$-invariant vector field by
$$\widetilde{\xi}\,  = \, \wtxi_H  + \wtxi_V \quad \in \quad \Gamma_P(HP)\oplus \Gamma_P(VP)\, ,$$
and define 
$$ \widetilde{\lambda}^\omega_{\xi}\, :=\, \iota_{\widetilde{\xi}} \widetilde{\om} \equiv \iota_{\widetilde{\xi}_V} \widetilde{\om} \quad \in \quad \Omega^0(P;\, \mathfrak{g})_{G\mbox{-}\mathrm{equiv}}\,  . $$

Denoting its corresponding ``infinitesimal gauge transformation''
on the base spacetime, via \eqref{VectorBundleValuedAreFormsTotalSpaceHorizontalForms}, by
$$ \lambda^\omega_{\xi} \quad \in \quad \Omega^0(X;\, P\times_G \mathfrak{g}) \, , 
$$
then the corresponding covariant Lie derivatives (Def. \ref{CovariantLieDerivative}) of the spacetime fields $\phi$ and $\omega$ on $X$ are given, equivalently, by
\begin{align}\label{CovariantLieDerivativerOfMatterFieldOnBase}
\widetilde{L}_{\xi}\phi \, &= \, \iota_\xi \dd^\omega \phi + \dd^\omega(\iota_\xi \phi) - \lambda_\xi^\omega \cdot \phi\\
&= \, L_\xi \phi + (\iota_\xi \omega - \lambda_\xi^\omega ) \cdot \phi \nonumber \, ,
\end{align}
where the latter term employs the induced action of the adjoint Lie algebra $P\times_G \frg $ on the associated vector bundle $V_G$ \eqref{AssociatedVectorBundle}, 
and
\begin{align}\label{CovariantLieDerivativerOfConnectionOnBase}
\widetilde{L}_{\xi}\omega \, &= \, \iota_\xi F_\omega + \dd^\omega( \lambda_\xi^\omega) 
\\
&= L_\xi \om - \dd^\omega(\iota_\xi \omega - \lambda_\xi^\omega)
\,.
\nonumber 
\end{align}
\end{lemma}
\begin{proof}
By \eqref{PullingBackTotalSpaceCovariantLieDerivative} the covariant Lie derivative expressions on X are given by pulling back those of the total space P from \eqref{TotalSpaceLieDerivativeMatterField} and \eqref{TotalSpaceLieDerivativeConection}, both being expressed as a sum of individually horizontal terms. The first equalities of the statement follow by the horizontality and verticality of the forms contracting the lifted vector field $\wtxi=\wtxi_H + \wtxi_V$, namely
\begin{align*}
\iota_{\wtxi} \wtom &= \iota_{\wtxi_H} \wtom + \iota_{\wtxi_V} \wtom = 0 + \iota_{\wtxi_V} \wtom  \\
&=: \widetilde{\lambda}_\xi^\omega 
\end{align*}
while 
\begin{align*}
\iota_{\wtxi} \widetilde{\phi} = \iota_{\wtxi_H} \widetilde{\phi} \qquad ,  \qquad \iota_{\wtxi} \dd^{\widetilde{\om}} \widetilde{\phi} = \iota_{\wtxi_H} \dd^{\widetilde{\om}} \widetilde{\phi} \quad \quad \mathrm{and} \qquad 
 \iota_{\wtxi} F_{\wtom} = \iota_{\wtxi_H} F_{\wtom} \, . 
\end{align*}
The latter equalities follow by expanding the resulting base spacetime covariant derivatives and curvature formulas, and then identifying the Cartan formula for the traditional Lie derivative $L_\xi= \dd \iota_\xi + \iota_\xi \dd$.
\end{proof}
It is worthwhile to expand on the two different expressions of the covariant Lie derivatives detailed above. Each term in the first equalities of \eqref{CovariantLieDerivativerOfMatterFieldOnBase} and \eqref{CovariantLieDerivativerOfConnectionOnBase} is a globally defined section of the corresponding bundle, hence guaranteeing the global existence of the covariant Lie derivative -- being simply their linear combination. On the other hand, each term in the second equalities of \eqref{CovariantLieDerivativerOfMatterFieldOnBase} and \eqref{CovariantLieDerivativerOfConnectionOnBase} is non-covariant, hence (individually) only defined locally. Nevertheless, they are so precisely such that the latter terms exactly cancel the inherent non-covariance of the traditional Lie derivative. In more detail, the ``correction'' term 
\begin{align}\label{CovariantizationTermSum}
B_{\xi}^\omega\, :=\, \iota_\xi \omega - \lambda_\xi^\omega
\end{align} 
is not globally defined on $X$, but rather represents a family of locally defined 0-forms
$$
\big\{(B_{\xi}^\omega)_i \quad \in \quad  \Omega^0(U_i; \, \mathfrak{g})\big\}_{i\in I}
$$
over some trivializing cover $\{U_i \hookrightarrow X\}$ of the bundle, linear in $\xi$, and related on overlaps as
$$
(B^\omega_\xi)_j= \gamma_{ji} \cdot (B^\omega_\xi)_i \cdot \gamma_{ij} + \gamma_{ji} \cdot \iota_\xi \differential \gamma_{ij} \, .
$$

This justifies naming the construction as a ``covariantization'' of the traditional Lie derivative from the perspective of the base spacetime X. Of course, since Def. \ref{CovariantLieDerivative} of the covariant Lie derivative does not involve any choice of connection, the following is immediate.
\begin{corollary}[\bf Covariantization term is independent of connection]\label{CovariantizationTermIsIndependentOfConnection}
Given an $\FR$-linear lift $\widetilde{(-)}\, : \, \mathcal{X}(X) \longrightarrow \mathcal{X}(P)_G$, the associated covariantization correction term \eqref{CovariantizationTermSum} 
of the traditional Lie derivative in Eqs. \eqref{CovariantLieDerivativerOfMatterFieldOnBase} and \eqref{CovariantLieDerivativerOfConnectionOnBase} on the spacetime $X$ is independent of the chosen connection, in that
$$
\iota_\xi \omega - \lambda^\omega_\xi \, = \iota_\xi \hat{\omega} - \lambda^{\hat{\omega}}_\xi \, . 
$$
\end{corollary}

\begin{example}[\bf Natural lift of chosen connection]\label{LiftingViaChosenConnection}
Given a connection $\wtom \in \Omega^1(P;\, \mathfrak{g})$, there exists a uniquely associated \textit{horizontal} lift (see e.g. Prop. 28.6 \cite{Tu17})
\begin{align*}
\widetilde{(-)}_{\widetilde{\om}}\, : \, \mathcal{X}(X) &\longrightarrow \Gamma_P(HP)\hookrightarrow \mathcal{X}(P)_G \\
\xi &\longmapsto \widetilde{\xi}_H + 0 \, .
\end{align*}
With this choice of lift, one has $\lambda_\xi^\omega=0$ and so the associated covariant Lie derivatives from Lem. \ref{CovariantLieDerivativeOnThebAseViaAConnection} take the form
\begin{align*}
L^{\om}_\xi \phi \, :&= \,  \iota_\xi \dd^\omega \phi + \dd^\omega(\iota_\xi \phi) \\
&=\,  L_\xi \phi + \iota_\xi \omega \cdot \phi \nonumber
\end{align*}
and
\begin{align*}
L^{\om}_\xi \om \, :&= \,  \iota_\xi F_\omega \\
&=\, L_\xi \omega - \dd^\omega(\iota_\xi \omega)\, ,  
\end{align*} thus recovering the  ``\textit{gauge-covariant (infinitesimal) coordinate transformation}'' or ``\textit{natural covariant Lie derivative}'' from \cite{Ja80}\cite{VF12} and \cite{OG06}\cite{CGRS20}, respectively. Notice, since the curvature of the connection $F_{\widetilde{\omega}}$ measures the non-integrability of horizontal distribution, $[\widetilde{\xi_1}_H,\, \widetilde{\xi_2}_H] = \widetilde{[\xi^1 ,\,  \xi^2]}_H - \iota_{\widetilde{\xi_1}_H} \iota_{\widetilde{\xi_2}_H} F_{\widetilde{\omega}}$, it follows that the corresponding covariant Lie derivatives define Lie algebra actions if any only if the connection is \textit{flat}.

Moreover, it is immediately apparent that an arbitrary covariant Lie derivative \eqref{CovariantLieDerivativerOfMatterFieldOnBase}\eqref{CovariantLieDerivativerOfConnectionOnBase}, expressed via a connection as in Lem. \ref{CovariantLieDerivativeOnThebAseViaAConnection}, is related to the corresponding natural covariant Lie derivative of the same connection via\footnote{In this most general case, the vanishing of the curvature is not sufficient for the covariant Lie derivatives to define Lie algebra actions. Instead, a straightforward (but lengthy) calculation shows that $[\widetilde{L}_{\xi_1}, \, \widetilde{L}_{\xi_2}] = \widetilde{L}_{[\xi_1, \xi_2]}$ if and only if $$
\iota_{\xi_1} \dd^\omega \lambda_{\xi_2} - \iota_{\xi_2} \dd^\omega \lambda_{\xi_1} -\lambda_{[\xi_1,\xi_2]} + [\lambda_{\xi_1},\lambda_{\xi_2}] -\iota_{\xi_1}\iota_{\xi_2} F_{\omega} \, = \,  0 \quad \in \quad \Omega^0(X,\, P\times_G \frg) \, .
$$}
\begin{align}\label{ArbitraryCovLieVsNaturalCovLieOnMatterField}
\widetilde{L}_{\xi} \phi = L_\xi^\omega \phi - \lambda^\omega_\xi \cdot \phi  
\end{align}
and
\begin{align}\label{ArbitraryCovLieVsNaturalCovLieOnConnection}
\widetilde{L}_{\xi} \omega = L_\xi^\omega \omega + \dd^\omega (\lambda^\omega_\xi)
\,.
\end{align}
\end{example}

Collecting all the above, we have the following equivalent ways of determining a covariant Lie derivative directly on a spacetime $X$, recovering the physically motivated formulas from \cite{OG06} and further refining their geometrical origin.
\begin{proposition}[\bf Equivalent ways of defining a covariant Lie derivative]\label{EquivalentWaysOfDefiningACovariantLieDerivative}
In addition to Def. \ref{CovariantLieDerivative} of the covariant Lie derivative, this may be equivalently defined directly on the spacetime $X$ via a choice of: 
\begin{itemize}
\item Either a family of $\FR$-linear maps
$$
\big\{B_{i}: \CX(X) \longrightarrow \Omega^0(U_i; \, \mathfrak{g})\big\}_{i\in I}
$$
over some trivializing cover $\{U_i \hookrightarrow X\}$ of the bundle, related on overlaps as
\begin{align}\label{CorrectionTermPatching}
B_j= \gamma_{ji} \cdot B_i \cdot \gamma_{ij} + \gamma_{ji} \cdot \iota_\xi \differential \gamma_{ij} \, ,
\end{align} 
yielding the associated covariant Lie derivative via
$$\widetilde{L}_\xi \phi = L_\xi \phi + B_\xi \cdot \phi $$
and
$$
\widetilde{L}_\xi \omega = L_\xi \omega - \dd^\omega B_\xi \, .
$$
\item Or a pair $(\om, \lambda)$ where $\om$ is a given by family of local  gauge fields $\{\om_i\}_{i\in I}$  on X representing a connection $\wtom$ on P and 
$$\lambda \, : \, \CX(X) \longrightarrow \Omega^0(X;\, P\times_G \mathfrak{g})$$ is an $\FR$-linear map from spacetime vector fields into sections of the adjoint bundle, with any two such pairs being considered equivalent if 
$$
\iota_\xi \omega - \lambda_\xi \, = \iota_\xi \hat{\omega} - \hat{\lambda}_\xi \, 
$$
for all $\xi \in \CX(X)$.
\end{itemize}
The latter approach is related to the former via 
$B_\xi = \iota_\xi \omega - \lambda_\xi$, which then recovers the form from \eqref{CovariantLieDerivativerOfMatterFieldOnBase} and \eqref{CovariantLieDerivativerOfConnectionOnBase}. Finally the pairs of globally defined forms $(\widetilde{\om}, \widetilde{\lambda})$ on $P$ uniquely determine the corresponding $\FR$-linear lifts 
$$
\widetilde{(-)}\, : \, \mathcal{X}(X) \longrightarrow \mathcal{X}(P)_G \, . $$

\end{proposition}
\begin{proof}
The fact that Def. \ref{CovariantLieDerivative} implies the spacetime formulation in terms of such a family $$
\big\{B_{i}: \CX(X) \longrightarrow \Omega^0(U_i; \, \mathfrak{g})\big\}_{i\in I}
$$ follows from the formulas in Lem. \ref{CovariantLieDerivativeOnThebAseViaAConnection} and Eq. \eqref{CovariantizationTermSum}. That this may be equivalently expressed exactly in terms of equivalency classes of pairs $(\om,\lambda)$ follows by choosing an arbitrary connection on $P$ (which always exist) and then defining $\lambda_\xi = \iota_\xi \om - B_\xi$.  Notice that a different choice of connection $\hat{\omega}$ yields $\hat{\lambda}_\xi= \iota_\xi \hat{\omega} - B_\xi$, so that in particular $\iota_\xi\omega -\lambda_\xi =\iota_\xi \hat{\omega} - \hat{\lambda}_\xi$, as per Cor. \ref{CovariantizationTermIsIndependentOfConnection}.
Lastly, the total space avatar of any such pair $(\wtom, \widetilde{\lambda})$ fully determines the lift 
$$
\widetilde{\xi} = \widetilde{\xi}_H + \widetilde{\xi}_V
$$
by assigning the horizontal component to be the unique horizontal lift wrt $\wtom$ (cf. Ex. \ref{LiftingViaChosenConnection}) and the vertical component to be the inverse image of 
$\widetilde{\lambda}_\xi \in \Omega^0(P;\, \mathfrak{g}) \cong \Gamma_P(P\times \mathfrak{g})$ 
under the bundle isomorphism induced by the given connection
$$
\widetilde{\omega}|_{VP} \, : \, VP \xlongrightarrow{\sim}  P\times \mathfrak{g} \, .
$$

\end{proof}

\begin{remark}[\bf $C^\infty(X)$-linear lifts]\label{CinftyVsRlinearity}
Those lifts $\widetilde{(-)}\, : \, \mathcal{X}(X) \longrightarrow \mathcal{X}(P)_G$ that are moreover $C^\infty(X)$-linear -- and not only $\FR$-linear (cf. Def. \ref{CovariantLieDerivative}) -- correspond  precisely to the horizontal lifts of connections from Ex. \ref{LiftingViaChosenConnection}. Indeed,  by Prop. \ref{EquivalentWaysOfDefiningACovariantLieDerivative} such a lift corresponds equivalently to a family $
\big\{B_{i}: \CX(X) \longrightarrow \Omega^0(U_i; \, \mathfrak{g})\big\}_{i\in I}
$, each member of which is also $C^\infty(X)$-linear, and so may be canonically identified with a family of local 1-forms $
\big\{B_{i}\, \in \,  \Omega^1(U_i; \, \mathfrak{g})\big\}_{i\in I}$ which patch via \eqref{CorrectionTermPatching}, hence being local gauge fields representing some connection $\widetilde{B}$ on $P$. Further presenting this by a pair $(\omega,\lambda)$ as in Prop. \ref{EquivalentWaysOfDefiningACovariantLieDerivative} implies that $\lambda$ is also necessarily  $C^\infty(X)$-linear, hence being identified with a $1$-form valued in the adjoint bundle $\lambda \in \Omega^1(X; \, P \times_G \frg)$. The relation $B= \omega - \lambda$ then reduces to the fact that connections on a principal $G$-bundle form an affine space modeled $\Omega^1(X; \, P \times_G \frg)$.

Crucially, however, such $C^\infty(X)$-linear lifts do not exhaust all relevant examples. Indeed, the Kosmann lift (Prop. \ref{TheKosmannLiftExists}), which is relevant to isometries of (super-)gravitational backgrounds, is an example of a lift which is only $\FR$-linear.
\end{remark}

\section{The covariant Lie derivatives of (Super-)Gravity}\label{CovariantLieDerivativesOfSupergravitySec}
Having laid out the general theory of covariant Lie derivatives, we now specialize to the case of (super-)gravitational theories in the coframe (vielbein) formalism. In the (super-)gravity literature, there are essentially two  physically (implicitly) established choices of covariant Lie derivatives entering the study of these theories. Namely, the \textit{natural covariant Lie derivative} (Ex. \ref{LiftingViaChosenConnection}) suitable for studying the relation of the (super-)diffeomorphism invariance of the theory to  local translational (super-)symmetry (see also \cite{CDF91}\footnote{We stress that the nomenclature in \cite{CDF91} is non-standard. In particular, the term ``soft group manifold'' is used for our corresponding principal G-bundle. Furthermore, the choice of covariant Lie derivative and the corresponding lift of vector fields is fixed (only) \textit{implicitly} to be that of Ex. \ref{LiftingViaChosenConnection}, and is done so in local trivializations of the bundle.} and \cite{EEC12}), and the Kosmann Lie derivative suitable for correctly identifying conserved currents due to background isometries in the coframe formalism and for acting appropriately on associated (spinorial) fields \cite{FF09}\cite{JM15}\cite{Pr17}\cite{BO23}\cite{BMO24}\cite{BMO25}. Furthermore, using the vanishing of the gauge covariant Kosmann Lie derivative on coframes as the correct notion of isometry we show that, in the case of abelian symmetries, the vanishing of the traditional Lie derivative is an equally consistent condition -- even on non-trivial topologies. Here we present facts and formulas regarding these in modernized and rigorous mathematical language. We do this with brief justifications, but omitting the long accompanying calculations when they already exist in the literature. The aim of this section is to disentangle the remaining confusion about these concepts in a straightforward manner for both mathematicians and theoretical physicists. In this section we shall distinguish super-manifolds from purely bosonic manifolds by the notation $X$ and $\bosonic{X}$, respectively, following the conventions of \cite{GSS24-SuGra}\cite{GSS-M5Brane}.
\subsection{(Super-)Diffeomorphism symmetry in (super-)gravity and local translational (super-)symmetry}\label{SuperDiffeomorphismSymmetrySec}
\subsubsection{Diffeomorphism vs local translational symmetry}

Recall, a (pure) gravitational field configuration on a bosonic manifold $\bosonic{X}$ in the first-order formalism is given by a pair $(e,\om)$, where $e$ is a coframe $e\in \Omega^{1}(\bosonic{X}; \, \FR^{1,d}_{\mathrm{SO}(1,d)})$ and $\omega$ is Lorentzian connection respectively, defined with respect to a background principal $\mathrm{SO}(1,d)$-bundle $\bosonic{P}\rightarrow 
\bosonic{X}$ -- the latter henceforth to be referred to as an (\textit{abstract}) ``$\mathrm{SO}(1,d)$-structure''.\footnote{This nomenclature is related, but is in slight clash, with the standard usage of the term \textit{(classical)} $G$-structure in the differential geometry literature, whereby it is used to refer to $G$-principal \textit{subbundles} of the \textit{frame bundle} $F\bosonic{X}$. See Rem. \ref{AnotherGeometricCharacterization} for more on the relation of the two notions in the current setting of $\mathrm{SO(1,d)}$-structures and the metric vs coframe formalisms in field theory.} The natural covariant Lie derivative from Ex. \ref{LiftingViaChosenConnection} appears naturally when considering the diffeomorphism (gauge) symmetry of the theory and its relation to the would-be translational gauge symmetry. The main point here is that \textit{on-shell} configurations of $D=1+d$ dimensional Einstein--Cartan gravity with Lagrangian 
\begin{align}\label{EinsteinCartanLagrangian}
\mathcal{L}_{\mathrm{EC}}(e,\omega) \, := \, e^{d-1}\wedge R_\omega = \epsilon_{a_0\cdots a_d}e^{a_0}\wedge\cdots \wedge e^{a_{d-2}}\wedge R^{a_{d-1} a_{d}}  \quad \in \quad \Omega^{d+1}(\bosonic{X})\, ,
\end{align}
have, in particular, vanishing torsion\footnote{Hence once solved for the Levi--Civita connection, one passes to the second-order formulation of general relativity.}
$$
T \, := \, \dd^\omega e \equiv 0 \, ,
$$
and as such the natural covariant derivative of the coframe, acting as the (off-shell) gauge symmetry generated by spacetime vector fields, reduces to
\begin{align*}
L_\xi^\omega e \, &= \, \iota_\xi T + \dd^\omega (\iota_\xi e) \\
&= \, \dd^\omega(\tau_\xi) 
\end{align*}
for 
$$
\tau_\xi := \iota_\xi e \quad \in \quad \Omega^{0}(\bosonic{X}; \, \FR^{1,d}_{\mathrm{SO}(1,d)})\, ,
$$
a ``\textit{translational gauge parameter}''.

Since the coframe field constitutes (by its non-degeneracy) an isomorphism of bundles\footnote{Note that the associated vector bundle $\FR^{1,d}_{\mathrm{SO}(1,d)}$ (often termed the ``fake tangent bundle'') is \textit{not} necessarily trivial, and hence the existence of an isomorphism to the tangent bundle over $\bosonic{X}$ does not require the latter to be parallelizable.} (i.e., a ``soldering form'') 
\begin{equation}\label{CoframeDef}
e: T\bosonic{X} \xrightarrow{\quad \sim\quad } \FR^{1,d}_{\mathrm{SO}(1,d)}\, ,
\end{equation}
it follows that on-shell 
(and only on-shell) 
the infinitesimal diffeomorphism transformation on the coframe, generated by vector fields via the natural covariant Lie derivative, may be equivalently interpreted as a local $\FR^{1,d}$-translational  gauge transformation on the coframe. This suggests that the combined field $(e,\om)$ could represent an actual connection for the Poincar{\'e} group $\mathrm{ISO}(\FR^{1,d}) := \mathrm{SO}(1,d)\ltimes \FR^{1,d}$, which is of course not the case.
\begin{remark}[\bf Einstein gravity is not a gauge theory for the Poincar\'{e} group]\label{EinsteinGravityIsNotPoincareGaugeTheory}
By Ex. \ref{LiftingViaChosenConnection}, the full field transformation along vector fields is given by
$$
\delta^{\mathrm{cov}}_\xi(e,\om)\, = \, (\, \iota_\xi T + \dd^\omega (\iota_\xi e),\, \iota_\xi R_\omega)\, ,
$$
which is indeed a symmetry of the Einstein--Cartan Lagrangian \eqref{EinsteinCartanLagrangian} (see e.g. \cite[p. 38]{CGRS20}\cite[p. 147]{CDF91}), 
and hence also of its field equations \cite[Prop. 2.42]{GS25}. On the subspace of on-shell fields, where in particular $T=0$, this transformation simplifies as
$$
\delta_\xi^{\mathrm{cov}}(e,\om)|_{\mathrm{on}\mbox{-}\mathrm{shell}}\,=  \,( \dd^\omega (\iota_\xi e),\, \iota_\xi R_\omega)\, ,
$$
remaining a symmetry of both the (restricted) Lagrangian and field equations. However, we stress that the transformation on $\omega$ is \textit{not the translational} part corresponding to the adjoint representation of $\mathrm{ISO}(d,1)$, which would yield instead
\begin{align}\label{OnShellTranslationalTranformation}
\delta_{\tau_\xi}^{\mathrm{gauge}}(e,\om)|_{\mathrm{on}\mbox{-}\mathrm{shell}}\,=  \,( \dd^\omega \tau_\xi,\, 0)\, .
\end{align}
for $\tau_\xi = \iota_\xi e$. Indeed, it is easy to see that this local translational transformation is not a symmetry of the Lagrangian \eqref{EinsteinCartanLagrangian}, nor of its field equations \cite[p. 147]{CDF91}.\footnote{The $D=1+2$ case is a famous exception, since in that case being on-shell also implies the flatness condition $R_\om=0$. In fact the $D=1+2$ theory is formally closely related to Chern-Simons theory, but due to the non-degeneracy condition on the coframe, it is not exactly equivalent to it \cite{Witten07}\cite{CGRS20}.}

Thus, it is incorrect to say that Einstein--Cartan gravity is a gauge theory for the Poincar{\'e} group, even though the composite pair $(e,\om)$ may be seen locally as a 1-form valued in the Poincar{\'e} Lie algebra. The coframe's non-degeneracy (soldering) condition, coupled with the natural gauge invariance of the Lagrangian \eqref{EinsteinCartanLagrangian} under only the Lorentz subgroup $\mathrm{SO
}(1,d) \hookrightarrow \mathrm{ISO}(1,d)$, establish \textit{decisively} Einstein gravity as a dynamical theory of \textit{Cartan connections} (see \cite{Sharpe}\cite{Catren15}\cite{Schreiber16}\cite{Scholz19}\cite{McKay}\cite{FR24} for modernized reviews and pointers to historical background literature, following \cite{Cartan23}). The extent to which the infinitesimal diffeomorphism symmetry of the Einstein--Cartan Lagrangian may be identified with infinitesimal local translational $\FR^{1,d}$-transformations is limited only to the coframe part of the field content, and furthermore only for on-shell configurations of the theory.
\end{remark}

\subsubsection{Super-diffeomorphism symmetry vs local translational supersymmetry}

Nevertheless, the upshot of this partial identification of the on-shell gauge symmetries is that it extends to theories of supergravity, albeit in their super-spacetime formulation. Indeed, by essentially parsing the statement in reverse, it provides an explanation for the origin of (on-shell) supersymmetry transformations as nothing but the action of the natural covariant Lie derivatives along odd vector fields. 
 
 Consider, for instance, the Lagrangian for $N=1$, $D=1+3$ supergravity
\begin{align}\label{EinsteinCartanRSLangrangian}
\mathcal{L}_{\mathrm{EC}}^{RS} (e,\omega,\psi) \, := \, e^{2}\wedge R_\omega + 4 \big(\,\overline{\psi} \,\Gamma_{0123}\, \Gamma_a \,  \rho \big) \wedge e^a   
\end{align}
on a manifold equipped with a Spin$(1,3)$-structure via a (bosonic) principal bundle $\bosonic{P}\rightarrow \bosonic{X}$, where 
$$
\rho\, := \, \dd^\omega \psi
$$ is the covariant derivative of the gravitino (``Rarita-Schwinger'') field 
$$
\psi \quad \in \quad  \Omega^1_\mathrm{dR}(\bosonic{X}; \, \mathbf{4}^\odd_{\mathrm{Spin}(1,3)}) \, ,
$$ 
a fermionic field valued  in the vector bundle associated to (real Majorana) $[3/2]$ representation of $\mathrm{Spin}(1,3)$ . The Lagrangian is again invariant under infinitesimal diffeomorphisms via the natural covariant Lie derivative, and local Lorentz transformations due the equivariance of the pairing $\big(\,\overline{(\cdot)} \, (\cdot) \,\big): \mathbf{4}\times \mathbf{4} \rightarrow \mathbb{R}$. Crucially, \eqref{EinsteinCartanRSLangrangian} is furthermore invariant under the infinitesimal local supersymmetry transformations (see e.g. \cite[(III.6.62a)-(III.6.62c)]{CDF91})
\begin{align}\label{4DSugraSusyTransf}
\delta^{\mathrm{susy}}_{\varepsilon} e \, & := \, 2 (\,\overline{\varepsilon} \, \Gamma  \, \psi) \nonumber\\
\delta^{\mathrm{susy}}_\varepsilon \psi &:= \, \dd^\omega \varepsilon
\\
\delta^{\mathrm{susy}}_\varepsilon \omega^{ab} & := \,  2\epsilon^{abb_1b_2} \big( \overline{\varepsilon}\,\Gamma_{0123}\,\Gamma_{b_3} \,\rho_{b_1 b_2}\big) \,e^{b_3} 
+ 2\epsilon^{b_1b_2b_3 [a} \big(\,\overline{\varepsilon}\Gamma_{0123}\,  \Gamma_{b_1}\, \rho_{b_2 b_3} \big) \,e^{b]} \, , \nonumber
\end{align}
where 
$$
\varepsilon \quad \in \quad  \Omega^0_\mathrm{dR}(\bosonic{X}; \, \mathbf{4}^\odd_{\mathrm{Spin}(1,3)}) \, ,
$$ 
is an ``\textit{odd translational gauge parameter}'' valued in the corresponding spinorial representation. 
We note that on-shell, 
where in particular
$$
\dd^\omega e  - \big(\, \overline{\psi} \, \Gamma \, \psi \big) \equiv 0 \, ,
$$
the supersymmetry transformation on the spin connection may be written in a different form which nevertheless remains \textit{non-vanishing} \cite[III.2.A3]{CDF91}. 

The transformations of the coframe and gravitino $(e,\,\psi)$, i.e., the super-coframe, have precisely the form of a local infinitesimal gauge transformation for the odd translational part of the super-Poincar{\'e} group\footnote{See \cite[Sec. 3]{GSS24-HiddenGroup} for a rigorous treatment of such super-Lie groups.}
\begin{align}\label{N1D4SuperPoincareGroup}
\mathrm{ISO}({\mathbb{R}^{1,3\vert \mathbf{4}}})\, := \, \mathrm{Spin}(1,3)\ltimes \FR^{1,3\vert \mathbf{4}}\, = \,  \mathrm{Spin}(1,3)\ltimes \big(\FR^{1,3}\times \mathbf{4}^{\odd}  \big) \, ,
\end{align}
but \textit{not} for the $\mathrm{Spin}(1,3)$-connection, even on-shell, which would have been
$$
\delta^{\mathrm{gauge}}_\varepsilon \omega = 0 \, .
$$
This parallels Rem. \ref{LiftingViaChosenConnection}, where the on-shell ``local  $\FR^{1,3}$-translations $\tau_\xi$'' of the theory do not act via the adjoint action of $\mathrm{ISO}(\mathbb{R}^{1,3})$ on $\omega$, but rather via the corresponding covariant natural Lie derivative along $\xi \in \CX(\bosonic{X})$. Said otherwise, the discrepancy is explained by showing the origin of the would-be translational symmetry to be instead diffeomorphism invariance.

For supergravity, an analogous explanation for the origin of local (translational) supersymmetry exists, provided one passes to its (on-shell) \textit{rheonomic super-space formulation} \cite{CDF91}. Summarizing the construction described therein in modern mathematical language\footnote{As recently established rigorously in the analogous but more complicated 11D super-space supergravity case in \cite{GSS24-SuGra}.}: Let $\bosonic{X}$ be a 4-dimensional manifold equipped with a $\mathrm{Spin}(1,3)$-structure $\bosonic{P}\rightarrow \bosonic{X}$, and consider the canonical supermanifold extension $X$ associated to the spinor representation $\mathbf{4}$ as \cite[Ex. 2.13, 2.77]{GSS24-SuGra}
\begin{align*}
X \, := \, \bosonic{X} \,\big\vert\,  \mathbf{4}_{\mathrm{Spin(1,3)}}  = \bosonic{X}\,  \big\vert\, (\bosonic{P}\times_{\mathrm{Spin(1,3)}} {\mathbf{4}})\, , 
\end{align*}
with its bosonic body being the original spacetime manifold
$$
\iota \, : \, \bosonic{X} \longhookrightarrow X \, .
$$

It is then a non-trivial result \cite[III.3.6]{CDF91} that a super-spacetime (\cite[Def. 2.74]{GSS24-SuGra})\footnote{In particular, this consists of a $\mathrm{Spin}(1,3)$-structure on X, an associated $\mathrm{Spin}(1,3)$ connection $\omega^s$ and a super-coframe which encodes an isomorphism between the super-tangent bundle and the associated super-vector bundle  $(e^s,\psi^s): TX \xrightarrow{\quad \sim \quad} \FR^{1,3\vert \mathbf{4}}_{\mathrm{Spin}(1,3)}$.} 
$$
\big(X,\,(e^s, \,\psi^s,\,\omega^s)\big)
$$ with the sole condition that its (superized) field strengths have expansions in the super-coframe basis $(e^s,\, \psi^s)$ of the particular form \cite[III.3.74.a-c]{CDF91}
\begin{align}\label{SuperSpaceFieldStrengthsOnshellParametrization}
T^s\,&:=\, \dd^{\omega^s} e^s \, \equiv\,  \big(\, \overline{\psi^s}\, \Gamma \psi^s\,\big) \nonumber \\
\rho^s\, &:= \, \dd^{\omega^s} \psi^s \equiv  \frac{1}{2}\, \rho^s_{ab}\,  e^{s,a}\wedge e^{s,b} \\
R^{s, ab} \, &:= \dd^{\omega^s} \omega^{s, ab} \equiv \, \frac{1}{2}\, R^{s, ab}_{cd} \,  e^{s,c}\wedge e^{s,d} \, + \,  2\epsilon^{aba_1a_2} \big( \overline{\psi^s}\,\Gamma_{0123}\,\Gamma_{a_3} \,\rho^s_{a_1 a_2}\big) \,e^{s, a_3} 
+ 2\epsilon^{a_1a_2a_3 [a} \big(\,\overline{\psi^s}\Gamma_{0123}\,  \Gamma_{a_1}\, \rho^s_{a_2 a_3} \big) \,e^{s,b]} \nonumber \, ,
\end{align}
is precisely equivalent to 
\begin{itemize}
\item[\bf (i)]  the restriction 
$$
(e,\psi,\om)\,:=\, \iota^* (e^s, \, \psi^s, \, \om^s) \, \equiv  \,(e^s, \, \psi^s, \, \om^s) |_{\theta=0}
$$ satisfying the on-shell field equations associated to the $N=1$ $4D$ supergravity Lagrangian \eqref{EinsteinCartanRSLangrangian} on the bosonic spacetime$\bosonic{X}$, and

\item[\bf (ii)] the superized fields being uniquely determined as a ``rheonomic extension''\footnote{The full justification of this latter statement is due to the existence of super-normal coordinates for the underlying $\mathrm{Spin}$ structure \cite{McArthur84}\cite{Tsimpis04}\cite[Rem. 2.3]{GSS24-M5Embedding}.} of their bosonic spacetime $\bosonic{X}$ restriction.
\end{itemize}

\medskip
Concisely:
\begin{equation}
  \label{SuperSpaceRheonomyStatement}
 \colorbox{lightgray}
  {
     \def\arraystretch{1.2}
    \begin{tabular}{c}
      $(4\vert\mathbf{4})$-dimensional 
      super-spacetimes  $\big(
        X, (e^s, \psi^s, \omega^s)
      \big)$
      \\
      with field strengths of the form
\eqref{SuperSpaceFieldStrengthsOnshellParametrization}.
\end{tabular}
  }
  \hspace{.3cm}
  \Leftrightarrow
  \hspace{.3cm}
  \colorbox{lightgray}{
    \begin{tabular}{c}
      Solutions of $N=1$ 4D SuGra \eqref{EinsteinCartanRSLangrangian}
      \\
$(e,\psi,\om)$ on the bosonic body $\bosonic{X}$. \end{tabular}
    }
\end{equation}

\medskip 
\noindent
Given this fact and explicit formulas, the (on-shell) relation of the theory's infinitesimal local spacetime supersymmetry on $\bosonic{X}$ to infinitesimal super-diffeomorphisms on $X$ becomes apparent. Namely, by the super-coframe property of $(e^s\, ,\, \psi^s )$ one has in particular an isomorphism of bundles over $X$
$$
\psi^s \, : \, T^{\odd}X \xrightarrow{\quad \sim\quad } \mathbf{4}^{\odd}_{\mathrm{Spin}(1,3)} \, ,
$$
and hence any \textit{odd} vector field\footnote{Strictly speaking this statement may be taken to include all sections of $T^\odd X$, which might not in general be \textit{odd} in the $\mathbb{Z}_2$-sense.}
$$
\eta  \quad \in \quad \CX^{\odd}(X)
$$
may be identified with an ``odd-translational gauge parameter'' 
$$
\varepsilon_\eta \, : = \, \iota_\eta \psi^s  \quad \in \quad \Omega^{0,\odd}(X; \, \mathbf{4}^{\odd}_{\mathrm{Spin}(1,3)})
$$
and vice-versa.

It follows that the action of an odd vector field $
\eta$ on any on-shell super-space supergravity field configuration, via the natural covariant Lie derivative (Ex. \ref{LiftingViaChosenConnection}), is given by 
\begin{align}\label{OddVectorFieldTransfAsSusyTransf}
\delta_{\eta}^\mathrm{cov}e^s \vert_{\mathrm{on-shell}} &= \iota_{\eta}\dd^{\omega^s} e^s+ \dd^{\omega^s}\big(\iota_{\eta}e^s) = \iota_{\eta^s} \big(\, \overline{\psi^s}\, \Gamma\, \psi^s\,\big)
\nonumber \\
&= 2 \big(\,\overline{\varepsilon}_\eta\, \Gamma \,\psi \big) \,  \nonumber
\\
\delta_{\eta}^\mathrm{cov} \psi^s \vert_{\mathrm{on-shell}} \, &= \, \iota_{\eta}\dd^{\omega^s} \psi^s+ \dd^{\omega^s} (\iota_{\eta}\psi^s) = \iota_\eta \rho^s + \dd^{\omega^s}\varepsilon_\eta \\
&= \dd^{\omega^s} \varepsilon_\eta
\nonumber \\
\delta_{\eta}^\mathrm{cov}\omega^s \vert_{\mathrm{on-shell}} \, &= \iota_\eta R^s=   0 + \iota_\eta \, \bigg( 2\epsilon^{abb_1b_2} \big( \overline{\psi^s}\,\Gamma_{0123}\,\Gamma_{b_3} \,\rho^s_{b_1 b_2}\big) \,e^{s, b_3} 
+ 2\epsilon^{b_1b_2b_3 [a} \big(\,\overline{\psi^s}\Gamma_{0123}\,  \Gamma_{b_1}\, \rho^s_{b_2 b_3} \big) \,e^{s,b]}  \bigg) \nonumber \\
&=  2\epsilon^{abb_1b_2} \big( \overline{\varepsilon}_\eta\,\Gamma_{0123}\,\Gamma_{b_3} \,\rho_{b_1 b_2}\big) \,e^{b_3} 
+ 2\epsilon^{b_1b_2b_3 [a} \big(\,\overline{\varepsilon}_\eta\Gamma_{0123}\,  \Gamma_{b_1}\, \rho_{b_2 b_3} \big) \,e^{b]} \, , \nonumber
\end{align}
where we used the on-shell form of the superized field strengths \eqref{SuperSpaceFieldStrengthsOnshellParametrization}. Finally, restricting the formulas to the underlying bosonic spacetime 
$$
\iota \, : \, \bosonic{X} \longhookrightarrow X \, 
$$
yields precisely the local super-symmetry transformations from Eq. \eqref{4DSugraSusyTransf}. 

\begin{remark}[\bf Supergravities are not gauge theories for super-Poincar{\'e} groups]\label{SupergravitiesAreNotGaugeTheoriesForSuperPoincareGroups}
In extension of Rem. \ref{EinsteinGravityIsNotPoincareGaugeTheory}, the non-zero supersymmetry transformation \eqref{4DSugraSusyTransf} on the Spin connection $\omega$ shows that $N=1$, $D=4$ supergravity is not a gauge theory for the super-Poincar{\'e} group $\mathrm{ISO}(\mathbb{R}^{1,3\vert\mathbf{4}})$ from \eqref{N1D4SuperPoincareGroup}. This fact persists and, in fact, becomes more apparent in higher dimensional supergravities. 

For instance, in the (maximal) case of $D=11$ supergravity \cite{CJS78} not only is the supersymmetry transformation of the $\mathrm{Spin}(1,10)$-connection non-zero (see e.g. \cite[\S III.8.I]{CDF91})
$$
\delta_{\varepsilon}^{\mathrm{susy}} \omega \, \neq 0 \, ,  
$$
where 
$
\varepsilon \, \in \, \Omega^0_{\mathrm{dR}}\big(\bosonic{X}; \, \mathbf{32}^{\odd}_{\mathrm{Spin}(1,10)} \big)  
$ is a odd gauge parameter, but the existence of the $(G_4,\, G_7)$-field also modifies the supersymmetry transformation of the corresponding gravitino field 
$$
\psi \quad \in \quad \Omega^1_{\mathrm{dR}}\big(\bosonic{X}; \, \mathbf{32}^{\odd}_{\mathrm{Spin}(1,10)} \big) 
$$ 
to
\begin{align}\label{11DSusyTransformationOfGravitino}
\delta_{\varepsilon}^{\mathrm{susy}}\psi \, := \, \dd^\omega \varepsilon \, +  \,  \big( \tfrac{1}{6}
      \tfrac{1}{3!}
      (G_4)_{a \, b_1 b_2 b_3}
      \Gamma^{b_1 b_2 b_3}
      \,-\,
      \tfrac{1}{12}
      \tfrac{1}{4!}
      (G_4)^{b_1 \cdots b_4}
      \Gamma_{a \, b_1 \cdots b_4}\big) \cdot \varepsilon \, e^a
      \, .
\end{align}
In other words, even on the supersymmetry partner of the coframe $e$, the supersymmetry transformation does not act as an odd local gauge translation, i.e., as it would if the gravitational field $(e,\, \psi ,\omega)$ of 11D supergravity were an actual gauge field for the corresponding super-Poincar{\'e} group
\begin{align}\label{D11SuperPoincareGroup}
\mathrm{ISO}({\mathbb{R}^{1,10\vert \mathbf{32}}})\, := \, \mathrm{Spin}(1,10)\ltimes \FR^{1,10\vert \mathbf{32}}= \mathrm{Spin}(1,10)\ltimes \big(\FR^{1,10}\times \mathbf{32}^{\odd}  \big) \,, 
\end{align}
in which case
$$
\delta_\varepsilon^{\mathrm{gauge}}\psi \, = \, \dd^\omega \varepsilon \, .
$$

Nevertheless, the explanation for the origin of these local supersymmetry transformations (on all fields) as the action of odd vector fields on the rheonomic super-space formulation of $11D$ supergravity (\cite[Thm. 3.1]{GSS24-SuGra}, following \cite[\S III.8.5]{CDF91}), acting via the natural covariant Lie derivative of Ex. \ref{CovariantLieDerivative}, is still valid. In particular, on the canonical extended super-spacetime
\begin{align*}
X \, := \, \bosonic{X} \,\big\vert\,  \mathbf{32}_{\mathrm{Spin(1,10)}}  = \bosonic{X}\,  \big\vert\, (\bosonic{P}\times_{\mathrm{Spin(1,10)}} {\mathbf{32}})\, , 
\end{align*}
the on-shell superized gravitino field strength has a coframe expansion of the form (cf.  \cite[Eq. (127), (135), (163)]{GSS24-SuGra})
$$
\rho^s\, := \, \dd^{\omega^s} \psi^s \equiv  \frac{1}{2}\, \rho^s_{ab}\,  e^{s,a}\wedge e^{s,b}  \, +  \,  \big( \tfrac{1}{6}
      \tfrac{1}{3!}
      (G^s_4)_{a \, b_1 b_2 b_3}
      \Gamma^{b_1 b_2 b_3}
      \,-\,
      \tfrac{1}{12}
      \tfrac{1}{4!}
      (G^s_4)^{b_1 \cdots b_4}
      \Gamma_{a \, b_1 \cdots b_4}\big) \cdot \psi \, e^{s,a}
      \, . 
$$
Similarly to the $4D$ case described previously, this immediately implies that the on-shell action of the natural covariant Lie derivative along an odd vector-field $\eta \, \in \CX^{\odd}(X)$ is given by
\begin{align*}
\delta_{\eta}^\mathrm{cov} \psi^s \vert_{\mathrm{on-shell}} \, &= \, \iota_{\eta}\dd^{\omega^s} \psi^s+ \dd^{\omega^s} (\iota_{\eta}\psi^s) = \iota_\eta \rho^s + \dd^{\omega^s}\varepsilon_\eta \\
&= \dd^{\omega^s} \varepsilon_\eta \, +  \,  \big( \tfrac{1}{6}
      \tfrac{1}{3!}
      (G^s_4)_{a \, b_1 b_2 b_3}
      \Gamma^{b_1 b_2 b_3}
      \,-\,
      \tfrac{1}{12}
      \tfrac{1}{4!}
      (G^s_4)^{b_1 \cdots b_4}
      \Gamma_{a \, b_1 \cdots b_4}\big) \cdot \varepsilon_\eta \, e^{s,a}
      \, . 
\end{align*}
for $\varepsilon_\eta \, := \, \iota_\eta \psi^s  \in \Omega^0_{\mathrm{dR}}\big(X; \, \mathbf{32}^{\odd}_{\mathrm{Spin}(1,10)} \big)$ the corresponding odd gauge parameter on super-spacetime. Restricting to the bosonic spacetime along $\iota: \bosonic{X} \hookrightarrow X$ yields precisely the susy transformation of the gravitino \eqref{11DSusyTransformationOfGravitino}. The local susy transformations of all the other fields in 11D supergravity may be recovered in exactly the same manner via the on-shell super-space forms of their field strengths. 

\smallskip
These observations provide a fully rigorous technical justification for the insightful intuition and calculations of \cite{CDF91}. Summarizing, it is incorrect to say that supergravities are gauge theories of the corresponding super-Poincar{\'e} groups. A more precise statement is that they are dynamical theories of super-Cartan connections \cite{Eder21}\cite{Eder23}\cite{Ratcliffe22}\cite{GSS24-SuGra} corresponding to the subgroup inclusions $\mathrm{Spin}(1,d)\hookrightarrow \mathrm{ISO}(\mathbb{R}^{1,d\vert \mathbf{N}})$, and potentially of additional (higher) gauge fields\footnote{And potentially of further scalar fields in some dimensions.}, via their super-spacetime formulation. Apart from the standard $\mathrm{Spin(1,d)}$ -- and \textit{not} $\mathrm{ISO}(\mathbb{R}^{1,d\vert \mathbf{N}})$ --  local gauge-invariance of these theories\footnote{And the potential extra local gauge-invariance due to additional higher gauge fields.}, the (often complicated) infinitesimal local supersymmetry transformations on the underlying bosonic spacetimes $\bosonic{X}$ may be identified as the (on-shell) infinitesimal super-diffeomorphism invariance of the theory on the canonically associated super-spacetime extension  $X$.  
\end{remark}

\subsection{The Kosmann Lie derivative and isometries of (super-)spacetime backgrounds}\label{TheKosmannLieDerivativeAndIsometriesOfSuperSpacetimeBackgrounds}

Aside from the general covariance of (super-)gravity, parametrized infinitesimally by vector fields and acting via the natural covariant Lie derivative, vector fields are also used to describe infinitesimal isometries of fixed background solutions of any such theory. Recall that for a fixed pure gravitational background $(\bosonic{X},g)$ in the metric formalism, a vector field $\xi\in \CX(\bosonic{X})$ is an \textit{infinitesimal isometry} or \textit{Killing} if it preserves the metric
$$
\delta_\xi g \, := \, L_\xi g\, \equiv \, 0\, .
$$
When this background spacetime is coupled to other dynamical gauge theoretic fields unrelated to the metric, e.g. gauge fields such as the electromagnetic gauge field $A$, conserved currents for this isometry are computed via the natural covariant Lie derivative (Ex. \ref{LiftingViaChosenConnection})
$$
\delta_{\xi} A \, := \, \iota_\xi F_{A} 
$$
and analogously for associated matter fields \cite[Eq. 3.5]{Ja80}.

However, for fields associated to the $\mathrm{SO}(1,d)$-structure corresponding to the metric, or more generally a $\mathrm{Spin}(1,d)$-lift thereof -- such as coframe and spinorial (fermionic) fields -- the appropriate choice of covariant Lie derivative is different. Namely, the relevant choice has come to be called the ``Spinorial Lie derivative'' or the ``\textit{Kosmann Lie derivative}'', given originally in local coordinates by  Lichnerowicz \cite{Li63} for the case of \textit{Killing} vector fields and then generalized (in an ad-hoc manner) to arbitrary vector fields by Kosmann \cite{Ko72}. Since any such covariant Lie derivative is actually completely determined by the corresponding lift of spacetime vector fields, here we provide a completely natural way to identify it, irrespective of the existence of fermions but which nevertheless  recovers the Spinorial Lie derivative. We do so in the coframe formalism, hence defining first the lift to any associated (abstract) background principal  $\mathrm{SO}(1,d)$-bundle (cf. Rem. \ref{AnotherGeometricCharacterization}).

\begin{definition}[\bf Kosmann lift]\label{KosmannLift}
Let $e$ be a coframe \eqref{CoframeDef} on $\bosonic{X}$ with respect to an (abstract) background $\mathrm{SO}(1,d)$-structure $\bosonic{P}\rightarrow \bosonic{X}$, defining the corresponding metric
$$
g \, : = \langle e \, ,  e \rangle  \equiv \eta_{ab}\, e^a \, \otimes \, e^b \, .
$$
Its \textit{Kosmann lift} of vector fields 
$ ( - )^K  : \, \mathcal{X}(\bosonic{X}) \longrightarrow \mathcal{X}(\bosonic{P})$ is the one determined by the condition that its covariant Lie derivative on the coframe satisfies 
\begin{align}\label{KosmannCondition}
L_\xi g \, = 2 \langle L_\xi^K e \, ,  e \rangle \equiv \, 2\,  \eta_{ab}  \, (L_{\xi}^K e)^a \, \otimes \, e^b \, = \,  \,2\,   L_{\xi}^K e_b\, \otimes \, e^b\, . 
\end{align}
\end{definition}
The point of this defining equality is that the right hand side involves \textit{globally defined objects}, and not just the naive Lie derivative of $e$ -- which is not generally globally defined. In particular, this  condition implies that 
$$
L_{\xi} g = 0 \quad \iff \quad  L_{\xi}^{K} e = 0 \, . 
$$
This is a well-known property of the Kosmann Lie derivative of the coframe (see e.g. \cite[p. 40-41]{OG06}\cite[p. 5]{JM15}), guaranteeing that conserved currents of second-order and first-order formalisms of equivalent gravitational theories will coincide \cite{FF09}. In other words, by its very defining property, it is the correct choice of covariant Lie derivative so that the equivalence of first order gravity theories and their second order versions is manifest not only at the level of on-shell moduli spaces, but also at the level of conserved currents.

Here we provide an intrinsically covariant proof of its existence and uniqueness, which also allows to identify the corresponding lift of vector fields and its covariant Lie derivative on all associated fields via Prop. \ref{EquivalentWaysOfDefiningACovariantLieDerivative}.
\begin{proposition}
[\bf The Kosmann lift exists]\label{TheKosmannLiftExists}
The Kosmann Lie derivative of a coframe $e$ exists, is unique and is completely determined by $e$ itself. Explicitly, it is given equivalently by
\begin{align}\label{KosmannCovariantLieDerivativeFormula}
(L^K_\xi e)^a \, :&= \, \eta^{ad} \,  (\dd^{\omega^{\mathrm{LC}}} \iota_\xi e )_{(db)} \, e^b \\ 
&= \, \eta^{ad} \,  (L_\xi e)_{(db)}\, e^b \nonumber\, .
\end{align}
Yet equivalently, expressed as a covariantization correction to the standard Lie derivative (cf. Eq. \ref{CovariantLieDerivativerOfMatterFieldOnBase} and Prop. \ref{EquivalentWaysOfDefiningACovariantLieDerivative}), this means
\begin{align}\label{KosmannCovariantLieDerivativeFormulaAsCorrection}
L^K_\xi e \, = \, L_\xi e + B^K_\xi \cdot e
\end{align}
where 
\begin{align*}
(B^K_\xi)^{a}{}_b\, :&= \, (\iota_\xi \omega^{\mathrm{LC}} - \lambda_\xi^K)^a{}_b \\
 &= -\big( \eta^{ad}\, (L_\xi e)_{[bd]}\big) 
\end{align*} 
with $(\lambda^K_\xi)^{a}{}_b \, = \, \eta^{ad} \, (\dd^{\omega^{\mathrm{LC}}}\iota_\xi e)_{[bd]}\, .$

Furthermore, this completely determines the corresponding Kosmann lift 
$$
 ( - )^K  : \, \mathcal{X}(\bosonic{X}) \longrightarrow \mathcal{X}(\bosonic{P}) 
$$
and hence covariant Lie derivatives for all associated fields.  
\end{proposition}
\begin{proof}
First, recall that the Lie derivative of the metric $g=\eta_{ab} \,  e^a\otimes e^b$ may be computed in any local trivialization for the coframe field using the naive Lie derivative on the local $1$-form representatives on the right hand side, i.e.,
$$
L_\xi g \, = \, \eta_{ab}\,  L_\xi e^a \otimes e^b \, + \, \eta_{ab} \, e^a \otimes L_\xi e^b \, .
$$
Of course, the point here is that even if the Lie derivative on \textit{individual} coframes fields $e$ is \textit{not} defined globally, the symmetrised expression of the right hand side is a globally defined $(0,2)$-tensor on $\bosonic{X}$ (since the left hand side is globally defined).

Contracting the first entry with elements of the dual frame $\widehat{e} = (\widehat{e}_d)$ we have
$$
\iota_{\widehat{e}_d} (L_\xi g)\, =\, (L_\xi e)_{db}\, e^b + (L_\xi e)_{fd}\, e^f \, = \, 2 \, (L_\xi e)_{(db)} \, e^b \, , 
$$
or equivalently
$$
L_\xi g \, = \, 2 \, (L_\xi e)_{(db)} \, e^d \otimes e^b \, .
$$
This yields one of the expressions for the Kosmann Lie derivative  from \eqref{KosmannCovariantLieDerivativeFormula}, i.e., 
$$
(L^K_\xi e)^a \, = \, \eta^{ad} \,   (L_\xi e)_{(db)} \, e^b \nonumber\, ,
$$
which is hence defined \textit{globally} as a section $L_\xi ^K e \in \Omega^1(\bosonic{X}; \, \FR^{1,d}_{\mathrm{SO}(1,d)})$. 

To witness this more explicitly, notice that in the above expression with symmetrized indices as shown, one may replace the naive Lie derivative with \textit{any} $\mathrm{SO}(1,d)$-covariant Lie derivative. Indeed, consider any such as a covariantization-corrected Lie derivative (Prop. \ref{EquivalentWaysOfDefiningACovariantLieDerivative})
$$
\widetilde{L}_\xi e \, = \, L_\xi e + B_\xi \cdot e \, .
$$
The covariantization term is (individually) only \textit{locally} defined, but crucially valued in $\mathfrak{so}(1,d)$, so that by the antisymmetric properties of the latter Lie algebra we have in \textit{any} local trivialization
$$
(B_\xi \cdot e)_{(db)} \, := \,  (B_\xi \cdot e)_{(d}{\,}^a \, \eta_{|a|b)}\, = \, (B_\xi)^{a}{}_{(d} \, \eta_{|a|b)} \, = \, (B_\xi)_{(ab)} \, = \, 0 \, .
$$
Thus
$$
(L_\xi^K e)^a \, = \, \eta^{ad} \,   (L_\xi e)_{(db)} \, e^b \, = \,  \eta^{ad}\, (\widetilde{L}_\xi e)_{(db)} \, e^b \, ,
$$
for \textit{any} covariant Lie derivative on the right hand side.

For further concreteness, and also computational purposes, there is a natural choice for the latter: the natural covariant Lie derivative (Ex. \ref{LiftingViaChosenConnection}) of the corresponding Levi--Civita connection $\omega^{\mathrm{LC}}= \omega^{\mathrm{LC}}(e)$
$$
L_\xi^{\omega^{\mathrm{LC}}}e \,  = \,  \iota_{\xi}  \dd^{\omega^{\mathrm{LC}}} e + \dd^{\omega^{\mathrm{LC}}} \iota_\xi e\, .
$$
Recalling the defining property of the Levi--Civita connection, i.e., the vanishing torsion $T:=\dd^{\omega^\mathrm{LC}}e = 0$, this further simplifies to
$$
L_\xi^{\omega^{\mathrm{LC}}}e \,  = \, \dd^{\omega^{\mathrm{LC}}} \iota_\xi e\, ,
$$
which identifies the Kosmann covariant Lie derivative of the coframe (Def. \ref{KosmannLift}) via the first formula of \eqref{KosmannCovariantLieDerivativeFormula}, comprised fully of covariant terms
$$
(L^K_\xi e)^a \, := \, \eta^{ad} \,   (\dd^{\omega^{\mathrm{LC}}} \iota_\xi e )_{(db)} \, e^b\, .
$$

Next, to obtain the Kosmann Lie derivative as a covariantized correction of the standard Lie derivative, notice that
$$
(L_\xi^K e)_d - (L^{\omega^{\mathrm{LC}}}_\xi e)_d\,  =\,  - (\dd^{\omega^{\mathrm{LC}}} \iota_\xi e)_{[bd]} e^b    
$$
and so the correction term is given by
\begin{align*}
B^K_\xi \cdot e \, :&= \, L^K_\xi e - L_\xi e = L^K_\xi e -  L^{\omega^{\mathrm{LC}}}_\xi e + \iota_\xi \omega^{\mathrm{LC}}\cdot e\\
&= (\iota_\xi \omega^{\mathrm{LC}}\,- \lambda^K_\xi) \cdot e
\end{align*}
for $(\lambda^K_\xi)^{a}{}_b \, = \, \eta^{ad} \cdot (\dd^{\omega^{\mathrm{LC}}}\iota_\xi e)_{[bd]}$, as claimed in Eq. \eqref{KosmannCovariantLieDerivativeFormulaAsCorrection}, where in the second equality we used the formula the natural covariant derivative of $\omega^{\mathrm{LC}}$ from Ex. \ref{LiftingViaChosenConnection}. Computing the correction $B^K_\xi$ via the second formula of Kosmann Lie derivative \eqref{KosmannCovariantLieDerivativeFormula}, or equivalently substituting for the Levi-Civita connection in terms of the coframe, yields the second form of the correction term 
$$
(B^K_\xi \cdot e)^a = - \eta^{ad}\, (L_\xi e)_{[bd]} \, e^b \, .
$$

Finally, since the coframe field takes values in the fundamental and, in particular, \textit{faithful} representation of $\mathrm{SO}(1,d)$, this completely determines the correction terms $B^K_\xi$ on all local trivializations. Thus by Prop. \ref{EquivalentWaysOfDefiningACovariantLieDerivative}, this determines \textit{the Kosmann lift} of spacetime vector fields to $\mathrm{SO}(1,d)$-invariant vector fields and the corresponding covariant Lie derivatives for fields valued in any representation thereof.
\end{proof}
It is important to note that
the vanishing of the Kosmann Lie derivative of the coframe, viewed as a condition on the given vector field $\xi \in \CX(\bosonic{X})$, is nothing but the \textit{Killing equation} when transported along the isomorphism 
$$
e \, : \, T\bosonic{X} \xrightarrow{\quad \sim \quad} \mathbb{R}^{1,d}_{\mathrm{SO}(1,d)}
$$
given by the chosen coframe.

\begin{corollary}[\bf The vanishing of the Kosmann Lie derivative and the Killing equation]\label{TheVanishingofKosmannLieAndTheKillingEquation}
The ``affine'' version of the Levi-Civita connection $\omega^{\mathrm{LC}}$ defined on the tangent bundle 
may be recovered as
$$
\nabla^{\mathrm{LC}} \xi \, := \, e^{-1} \big( \dd^{\omega^{\mathrm{LC}}} e(\xi) \big) \quad \in \quad \Omega^{1}_\mathrm{dR} \big(\bosonic{X};\, T\bosonic{X} \big) \, , 
$$
with which it is immediate to see that 
\begin{align*}
L_\xi^K e = \, 0 \,  \quad &\iff \quad (\dd^{\omega^{\mathrm{LC}}} \iota_\xi e)_{(db)} \, = \, 0 \\
 &\iff \quad \nabla_{(d}^{\mathrm{LC}}\, \xi_{\, b)}\,  = \, 0 \, .
\end{align*}
\end{corollary}
\begin{remark}[\bf Lifting to $\mathrm{Spin}(1,d)$ covers]\label{LiftingToSpinCovers}
It is a standard fact that $G$-invariant vector fields on principal $G$-bundles lift canonically to any of their covering space principal bundle refinements (see e.g. \cite[Prop. 2.10]{GM03}). Thus the Kosmann lift $
 ( - )^K  : \, \mathcal{X}(\bosonic{X}) \longrightarrow \mathcal{X}(\bosonic{P})
$ canonically extends to a lift to any $\mathrm{Spin}(1,d)$-structure refinement
 $$
\widehat{\bosonic{P}}\longrightarrow \bosonic{P}
 $$
over $\bosonic{X}$.
 
From the point of view of the spacetime formulas for the Kosmann Lie derivative of the associated spinorial fields, this changes nothing apart from the patching of 
 \begin{align*}
B^K_\xi\, :&= \, (\iota_\xi \omega^{\mathrm{LC}} - \lambda_\xi^K) \, ,
\end{align*} 
with $(\lambda^K_\xi)^{a}{}_b \, = \, \eta^{ad} \cdot (\dd^{\omega^{\mathrm{LC}}}\iota_\xi e)_{[bd]}\, ,$ being via the corresponding $\mathrm{Spin}(1,d)$-valued transition functions. This allows then to consistently define the \textit{spinorial Kosmann Lie derivative} on associated spinor fields in either of their forms via
\begin{equation}
  \label{SpinorialLieDerivative}
  \begin{aligned}
  L_\xi^K \psi 
    \,&=\, 
  L_\xi \psi + B_\xi^K \cdot \psi 
   \,=\, 
  L_\xi^{\omega^\mathrm{LC}} \psi - \lambda_\xi^K \cdot \psi \nonumber 
  \\ 
  &= 
    \iota_\xi \dd^{\omega^\mathrm{LC}} \psi 
    + 
    \dd^{\omega^{\mathrm{LC}}} \iota_\xi \psi - \lambda_X^K \cdot \psi 
    \,,
  \end{aligned}
\end{equation}
with $B_\xi^K$ and $\lambda_\xi^K$ now acting via the corresponding $\mathrm{Spin}(1,d)$ representation. In particular, for $\psi_0$ a spinorial $0$-form in a standard $\Gamma$-matrix representation, the above reduces to 
\begin{align*}
L_\xi^K \psi_0 \, &=\,\iota_\xi \dd^{\omega^\mathrm{LC}} \psi_0 - \lambda_\xi^K \cdot \psi_0 \\
&= \, \xi^b (\dd^{\omega^\mathrm{LC}} \psi_0)_b - \tfrac{1}{4}(\dd^{\omega^{\mathrm{LC}}}\iota_\xi e)_{[bd]}\, \Gamma^b \Gamma^d \cdot \psi_0 
\end{align*}
whereby identifying the covariant derivatives with the corresponding affine ones as in Cor. \ref{TheVanishingofKosmannLieAndTheKillingEquation}, this recovers exactly the original formula of Lichnerowicz \cite{Li63} and Kosmann \cite{Ko72}.
\end{remark} 

\begin{remark}[\bf Characterization via (Gauge)-Natural bundles]\label{AnotherGeometricCharacterization}There exists an alternative geometrical characterization of the Kosmann lift when working directly with \textit{classical} $\mathrm{SO}(1,d)$-structures, i.e.,  $\mathrm{SO}(1,d)$-reductions of the full frame $\mathrm{GL}(d+1)$-bundle of the spacetime $\bosonic{X}$. That is, equivalently, when working directly with the metric $g$ and its $\mathrm{SO}(1,d)$-bundle of orthonormal frames 
$
F\bosonic{X}_g \longhookrightarrow F\bosonic{X}
$
(a.k.a. the second order formalism from a field theoretic gravitiational point of view). The details of this approach are laid out in \cite{FFFG96}\cite{GM03}\cite{FF03}, which we do not reproduce here as they are not directly relevant to the coframe formalism applications that we consider. 

A brief outline of this approach, starting from our coframe point of view, is as follows. One first identifies the background $\mathrm{SO}(1,d)$-bundle $\bosonic{P}\rightarrow \bosonic{X}$ with a subbundle of the full frame bundle, by considering the bundle of orthonormal frames of the corresponding metric $g= \eta_{ab}\, e^a\otimes e^b$, that is equivalently, by transferring the former along the isomorphism of the given coframe \eqref{CoframeDef}.\footnote{Since the coframe is an isomorphism of vector bundles $e : T\bosonic{X} \xrightarrow{\sim} \FR^{1,d}_{\mathrm{SO}(1,d)}$, its inverse  induces an (embedding) morphism between the associated $\mathrm{SO}(1,d)$-bundle $\bosonic{P}$ and the associated full $\mathrm{GL}(d+1)$-frame bundle $F\bosonic{X}$ of the tangent bundle, thus identifying the corresponding subbundle of orthonormal frames.} It can then be shown (\cite[Thm 4.12, Cor. 4.15]{GM03}) that for any such fixed reduction
$$\iota_g \, : \, F\bosonic{X}_g \longhookrightarrow \mathrm F\bosonic{X} \, ,
$$
there is a \textit{canonically associated}  splitting of the pullback bundle
$$
\iota_g^*\, T (F\bosonic{X}) \quad \cong \quad  TF\bosonic{X}_g \oplus \mathcal{M} 
$$
over $F \bosonic{X}_g$. 
Finally, using the natural \textit{prolongation lift} $\xi^N$ of any vector field $\xi$ on the base $\bosonic{X}$ to the \textit{full} frame bundle $F \bosonic{X}$ (being an example of a ``natural bundle'' over $\bosonic{X}$ \cite{KMS93}), one can show \cite[Ex. 4.17]{GM03} that the (local formula of the) corresponding Kosmann vector field $\xi^K$ is recovered precisely as the projection onto $T F\bosonic{X}_g$, with respect to the above decomposition
$$
\xi^K \, = \, \xi^N \vert_{TF\bosonic{X}_g} \, .
$$
\end{remark}

\subsubsection{Isometry in the coframe formalism and Kaluza--Klein dimensional reduction}\label{IsometryInTheCoframeFormalismAndKaluzaKleinDimensionalReduction}

The literature on dimensional reduction à la
Kaluza--Klein  of field theories in the coframe formalism of gravity has traditionally employed the vanishing of its naive Lie derivative along certain vector fields as the background symmetry condition. More precisely, the usual demand is that the coordinate basis components of the coframe are independent of the (local) adapted coordinate associated to the vector field generating the symmetry \cite{ScSc79}\cite{Cr81}\cite{DNP86}. Explicitly, say for the case of a trivial $S^1$-bundle spacetime $\widetilde{X} = \bosonic{X}\times S^1$ with coordinate $\theta$ along $S^1$, one requires that
$$
\partial_\theta \, 
e^a_\mu\,  = \,0
$$
where $e= (e^a_\mu \, \dd x^\mu)$ in a local coordinate chart for $\bosonic{X}\times S^1$. It is easy to see that this condition is indeed equivalent to the vanishing of the naive Lie derivative of the coframe along $\partial_\theta \in \CX(\bosonic{X}\times S^1)$
$$
L_{\partial_\theta} e \, = \, 0 \, .
$$
 
The issue is that this is not gauge-covariant statement (cf. Eq. \eqref{NonCovarianceLieDerivativeMatterField}) and so, a priori, does  not make sense globally on non-parallelizable base spacetimes $\bosonic{X}$ and furthermore on non-trivial principal $S^1$-bundles $\widetilde{X}$ over them. Nevertheless, here we make precise how this condition is indeed justified in precisely the situation of dimensional reduction along general \textit{abelian} $G$-fibers, such as the commonly studied cases of a circle $S^1\cong U(1)$ and tori $T^k\cong U(1)^{\times k}$. This is achieved by imposing the appropriate gauge covariant isometry condition as the vanishing of the Kosmann Lie derivative (Def. \ref{KosmannLift}, Cor. \ref{TheVanishingofKosmannLieAndTheKillingEquation}), and then exhibiting the existence of a special gauge in which this reduces to the vanishing of the traditional Lie derivative.
\begin{lemma}[\bf Invariance via Kosmann and naive Lie derivative of coframe]\label{InvarianceViaKosmannAndNaiveLieDerivativeOfCoframe}
Let $\widetilde{X}$ be a manifold with a free, abelian $G$-action, hence defining a principal $G$-bundle $\pi: \widetilde{X}\longrightarrow \bosonic{X}\cong \widetilde{X}/G$, further supplied with an $SO(1,d)$-structure and an associated coframe field $e  :  T\widetilde{X}\xlongrightarrow{\sim} \mathbb{R}^{1,d}_{\mathrm{SO}(1,d)} $ which is (infinitesimally) $G$-invariant via the Kosmann Lie derivative, i.e.,
$$
L_{A^\#}^K e \, = \, 0 
$$
for all fundamental vector fields $A^\# \in \Gamma_{\widetilde{X}}(V\widetilde{X})$ generated by the $G$-action, corresponding to  elements of the Lie algebra $A \in \frg = \mathrm{Lie}(G)$. Assuming the fundamental vector fields are spacelike for the corresponding metric 
$$
g \, = \, \langle  e , e \rangle \, = \, \eta_{ab}\, e^a\otimes e^b \, ,
$$
then:
\begin{itemize}
\item 
There always exists a gauge equivalent field configuration 
\begin{align*}
e \longmapsto \hat{e} \,= \, r \cdot e 
\end{align*}
for some $\mathrm{SO}(1,d)$-gauge transformation $r: \widetilde{X}\rightarrow \mathrm{Ad}(P) $
with respect to the corresponding $\mathrm{SO}(1,d)$-structure $P\rightarrow \widetilde{X}$,  for which the vanishing of the Kosmann Lie derivatives are equivalent to the vanishing of the naive Lie derivatives 
$$
L_{{A^\#}} \hat{e} \,  = \, 0 \, .
$$
\item For $\mathrm{dim}(G)=k$, there exists a gauge sub-equivalence class of coframes with the same property, given by the orbit of $\hat{e}$ under $G$-invariant local $\mathrm{SO}(1,d-k)\times \mathrm{SO}(k)$-transformations with respect to an accordingly reduced $\mathrm{SO}(1,d-k)\times\mathrm{SO}(k)$-structure $\hat{P}\rightarrow \widetilde{X}$.
\end{itemize}
\end{lemma}
\begin{proof}
Consider the metric corresponding to the (equivalence class of the) coframe
$$
g \, := \, \langle e\, ,  e \rangle \, = \, \eta_{ab}\, e^a\otimes e^b \, . 
$$
By construction \eqref{KosmannCondition}, the vanishing of the Kosmann Lie derivative of $e$ is equivalent to the (infinitesimal) $G$-invariance of the metric
$$L_{A^\#} g \, =\, 0 \, . $$
Since this infinitesimal $G$-action is assumed to originate from a finite right action $\rho: \widetilde{X}\times G\rightarrow \widetilde{X}$, it follows that
$$
\rho_{g'}^*\, g \, = \, 0 
$$
for all $g'\in G$.

It is a standard fact that a $G$-invariant metric on a principal $G$-bundle (spacelike along the $G$-fibers) defines a horizontal splitting of its tangent bundle by taking the horizontal vector fields to be those orthogonal to any vertical vector field 
$$
\Gamma(H\widetilde{X}):= \big\{ X_H\in \CX(\widetilde{X}) \, \,  | \, \, g(X_H,X_V)=0 \quad \forall \quad X_V \in \Gamma(V\widetilde{X})\big\}\, .
$$
The $G$-invariance of the metric guarantees that the resulting distribution is $G$-invariant, and so it follows that the metric has a Kaluza--Klein-like decomposition of the form
\begin{align}
g\, = \,\pi^*h+ g^V(\theta,\, \theta)
\end{align}
where $\theta \in \mathrm{Conn}_G(\widetilde{X})$ is the corresponding connection 1-form\footnote{Recall, given a $G$-invariant horizontal distribution on $P$, its connection 1-form is defined via $\theta(X_H)=0$ for all horizontal $X_H\in \Gamma(H\widetilde{X})$ and $\omega(A^\#)=A$ for all fundamental vector fields $A^\# \in \Gamma(V\widetilde{X})$.} $\theta = \theta^m\, E_m$ and $g^V=g^V_{mn} \, T^m\otimes T^n$ is a G-equivariant $(\frg^*\otimes \frg^*)$-valued function\footnote{Hence, via Eq. \eqref{VectorBundleValuedAreFormsTotalSpaceHorizontalForms}, corresponding to a section of the tensor product of the coadjoint bundle $\mathrm{Ad}^*(\widetilde{X})\otimes\mathrm{Ad}^*(\widetilde{X})$ over the dimensionally reduced base spacetime $\bosonic{X}=\widetilde{X}/G$.} (due to the $G$-invariance of $g$ and $G$-equivariance of $\theta$), where $\{E_{m}\}_{m=1,\cdots,k}$ and $\{T^m\}_{m=1,\cdots,k}$ are dual bases of $\frg$ and $\frg^*$ respectively, 
and $h \in \mathrm{Met}(\bosonic{X})$ is a metric on the base so that in particular
$$
L_{A^\#} \pi^* h \, = \, 0 \,. 
$$
Crucially, if $G$ is abelian it follows that the $G$-equivariance of $\theta$ and $g_V$, being via the adjoint and coadjoint representations, also reduce to $G$-invariance, i.e., 

$$
L_{A^\#} \theta \, =\, 0 \quad \quad \mathrm{and} \quad \quad L_{A^\#} 
 g_V \, = \, 0\,  \, .
$$

Applying the Gram--Schmidt (GS) algorithm on the  basis\footnote{For the corresponding sub-module of vertical 1-forms $\Omega^{1}_{\mathrm{Vert}}(\widetilde{X})\hookrightarrow \Omega^1(\widetilde{X})$  spanned by the connection 1-form components.} of $1$-forms given by the (globally defined) components $\{\theta^m\}$ of the  connection (with respect to the dual metric of $g$) 
$$
\mathrm{GS} \, : \, \{\theta^{m}\}_{m=1,\dots,k} \xmapsto{\quad \quad } \{\hat{e}^\alpha_\theta\}_{\alpha=1,\dots,k} 
$$
yields a new basis for  vertical $1$-forms, \textit{orthonormal} w.r.t $g_V$. Moreover, the GS algorithm induces a (fractional) polynomial functional dependence on the $G$-invariant $g_V$ and $\theta$ 
$$
e^1_\theta \, = \, \frac{\theta^1}{ |\theta^1|_g} \quad \quad , \quad \quad e^2_\theta \, = \,\frac{ \big(\theta^2 - \tfrac{g(\theta^2,\theta^1)}{g(\theta^1,\theta^1)}\cdot \theta^1\big)}{ \big|\big(\theta^2 - \tfrac{g(\theta^2,\theta^1)}{g(\theta^1,\theta^1)}\cdot\theta^1\big)\big|_g} \quad \quad , \quad \quad \cdots 
$$
which implies that the new vertical coframe $\hat{e}_\theta=(\hat{e}^0_\theta,\cdots, \hat{e}_\theta^{k})$ is also \textit{$G$-invariant} via the traditional Lie derivative
$$
L_{A^\#} \hat{e}_\theta \, = \, 0\, . 
$$
Choosing further a coframe 
$$
\hat{e}_h=(\hat{e}^0_h,\cdots, \hat{e}_h^{d-k}) \quad : \quad T\bosonic{X}\xrightarrow{\quad \sim \quad } \mathbb{R}^{1,d-k}_{\mathrm{SO}(1,d-k)}
$$ for the base metric $h$, it follows immediately that $\pi^* \hat{e}_h$ is also a (naively) $G$-invariant horizontal coframe  
$$
L_{A^\#} \pi^*\hat{e_h} \, = \, 0\, , 
$$
and so 
\begin{equation}\label{GaugeFixedInvariantCoframe}
\hat{e} \, := \, \big( \pi^*\hat{e}^0_h, \cdots, \, \pi^*\hat{e}^{d-k}_h, \, \hat{e}^0_\theta, \cdots , \, \hat{e}_\theta^k \big) 
\end{equation}
defines a total orthonormal coframe for $g$
\begin{align*}
g\, &=\, \langle \hat{e} \, , \, \hat{e} \rangle \, = \,  \eta_{ab}\, \hat{e}^a\otimes \hat{e}^b \\
&= \, \eta_{AB}  \, \pi^*\hat{e}_h^A\otimes \pi^*\hat{e}_h^B \, + \,  \delta_{\alpha\beta}  \, \hat{e}^\alpha_\theta \otimes \hat{e}^{\beta}_\theta
\end{align*}
whose traditional Lie derivative happens to vanish 
$$
L_{A^\#} \hat{e} \, = \, 0 \, ,
$$
along all fundamental vector fields $A^\#\in \Gamma(V\widetilde{X})$.

Next, any two orthonormal coframes of the same metric are related by a unique $\mathrm{SO}(1,d)$-gauge transformation $r: \widetilde{X}\longrightarrow \mathrm{Ad}(P)$(see e.g. \cite[Lem. 2.7]{GSS-M5Brane})
$$
e  \xmapsto{\quad r \quad } \hat{e} = r\cdot e \, , 
$$
and finally, notice that in the above one may choose \textit{any} $\mathrm{SO}(1,d-k)$-gauge equivalent horizontal coframe $\widetilde{e}_h$ for the base metric $h\in \mathrm{Met}(\bosonic{X})$ and \textit{any} $\mathrm{SO}(d-k)$-gauge equivalent, $G$-invariant, vertical coframe $\tilde{e}_\theta$. This completes the proof.
\end{proof}

The (partially) fixed gauge  for the coframe field $\hat{e}$ from Eq. \eqref{GaugeFixedInvariantCoframe} moreover guarantees that the $G$-invariance of any associated field may also be consistently expressed via the traditional Lie derivative. 
\begin{corollary}
[\bf Invariance of associated fields]\label{InvarianceOfAssociatedFields}
Consider the setup of Lem. \ref{InvarianceViaKosmannAndNaiveLieDerivativeOfCoframe}, supplied further with a $\mathrm{Spin}(1,d)$-lift of the $\mathrm{SO}(1,d)$-structure (Rem. \ref{LiftingToSpinCovers}). For any associated spinorial field $\psi$ and spin connection $\omega$ which are also $G$-symmetric, i.e.,
$$
(L^K_{A^\#}e, \, L^K_{A^\#} \psi,\, L^K_{A^\#} \omega) \,= \, 0 \,    
$$
for all fundamental vector fields $A^\# \in \Gamma(V\widetilde{X})$,
then in the adapted gauge where 
$$
L_{A^\#} \hat{e} \, = \, 0 \, ,
$$
it is also the case that\footnote{Implicitly, here one means the $\mathrm{Spin}(1,d)$-gauge transformation acting on $\psi$ is an arbitrary lift of the $\mathrm{SO}(1,d)$-transformation $r: e\mapsto \hat{e}$ w.r.t. the $\mathrm{Spin}(1,d)$-structure cover.}
$$
L_{A^\#}\hat{\psi} \, = \, 0 \quad \mathrm{and} \quad L_{A^\#}\hat{\omega} \, = \, 0\,  .
$$
\end{corollary}
\begin{proof}
By the covariance property of the Kosmann Lie derivative, it follows that in the adapted gauge described in Lem. \ref{InvarianceViaKosmannAndNaiveLieDerivativeOfCoframe} both 
$$
L^K_{A^\#} \hat{e} \, = r\cdot L^K_{A^\#} e =  \, 0 \quad \quad \mathrm{and} \quad \quad  L_{A^\#} \hat{e} \, = \, 0 \, 
$$
vanish. By the covariantization-correction formula \eqref{KosmannCovariantLieDerivativeFormulaAsCorrection} for the Kosmann Lie derivative, this implies that
$$
B^K_{A^\#}\cdot \hat{e} \, = \, 
 L^K_{A^\#}\hat{e} -L_{A^\#} \hat{e}   \, = \, 0
$$
and so by faithfulness of the fundamental $\FR^{1,d}$-representation 
$$ B^K_{A^\#}\, = \, 0 $$
in the given gauge, for all fundamental vector fields $A^\# \in \Gamma(V \widetilde{X})$. Thus, by Prop. \ref{EquivalentWaysOfDefiningACovariantLieDerivative} it follows that for any associated $G$-symmetric spinorial field, in the above gauge,
$$
L_{A^\#} \hat{\psi} \, = \, L^K_{A^\#}\hat{\psi} - B^K_{A^\#}\cdot \hat{\psi} = 0 - 0 = 0 \, . 
$$
Similarly for any $G$-symmetric spin connection in the same gauge
$$
L_{A^\#}\hat{\omega}  = \widetilde{L}^K_{A^\#} \hat{\omega}  + \dd^{\hat{\omega}} B^K_{A^\#} = 0 + 0 = 0 \, .
$$
\end{proof}

\begin{corollary}[\bf Dimensional reduced KK-field content of $G$-invariant coframe]\label{KK-FieldContentOfG-InvariantCoframe}
The partially gauge-fixed $G$-invariant coframe $\hat{e}$ of Lem. \ref{InvarianceViaKosmannAndNaiveLieDerivativeOfCoframe} from Eq. \eqref{GaugeFixedInvariantCoframe} encodes exactly the following fields on the dimensionally reduced spacetime $\bosonic{X}= \widetilde{X} / G$:
\begin{itemize}
\item[\bf (i)] A coframe 
$$\hat{e}_h \, : \, T\bosonic{X} \xrightarrow{\quad \sim \quad} \FR^{1,d-k}_{\mathrm{SO}(1,d-k)} $$
with respect to the induced (abstract) $\mathrm{SO}(1,d-k)$-structure on $\bosonic{X}$,
\item[\bf(ii)] A family of locally defined $G$-gauge fields
$$
\big\{\theta_i \quad  \in \quad  \Omega^1( U_i; \, \mathfrak{g}) \big\}_{i\in I}\, ,
$$
on a trivializing cover  $\{U_i \hookrightarrow \bosonic{X}\}_{i\in I}$ of the $G$-bundle, corresponding to the connection $1$-form $\theta\in \Omega^1(\widetilde{X};\,\frg)$,
\item[\bf (iii)] A globally defined scalar matter field valued in $\frg^*\otimes\FR^k$
$$
\Phi\, : \, \bosonic{X} \xrightarrow{\quad \, \quad } \frg^*\otimes \FR^k \, ,
$$
whose matrix of components $(\Phi^{\alpha}\,_{m})$ w.r.t. any basis for $\frg^*\otimes \FR^k$ has positive determinant $\mathrm{det}(\Phi)\in \FR^+$.
\end{itemize}
Moreover, each of these transforms in the obvious manner under local $G\times\mathrm{SO}(1,d-k)\times\mathrm{SO}(k)$-transformations.
\begin{proof}
The basic coframe content $\hat{e}_h$ of the total coframe $\hat{e}$ \eqref{GaugeFixedInvariantCoframe} is obvious by construction. For the latter two fields, notice that
the vertical coframe constitutes a parallelism of the vertical bundle over $\widetilde{X}\rightarrow \bosonic{X}$ with a trivial $k$-dimensional Euclidean bundle over $\widetilde{X}$
$$
\hat{e}_\theta \, : \, V\widetilde{X} \xrightarrow{\quad \sim \quad }\widetilde{X}\times \FR^{k}\, , 
$$
since each of its components are globally defined as $1$-forms on $\widetilde{X}$, while on the other hand the connection $1$-form yields an isomorphism with a trivial $\frg$-bundle over $\widetilde{X}$
$$
\theta \, : \, V\widetilde{X} \xrightarrow{\quad \sim \quad } \widetilde{X}\times \frg \, .
$$
Each of the vertical coframe components $\hat{e}^a_\theta$ may be expanded in the vertical basis of connection 1-form components
$$
\hat{e}^a_\theta \, = \, \widetilde{\Phi}^{a}\,_{m} \, \theta^m
$$
for some matrix component functions $(\widetilde{\Phi}^a\,_m)$, which are \textit{necessarily} $G$-invariant by the $G$-invariance of both $\theta$ and $\hat{e}_\theta$. Notice furthermore, since the GS-algorithm employed to obtain $\{\hat{e}_\theta^\alpha\}$ from $\{\theta^m\}$ preserves the orientation ($\equiv$ ordering) of the latter, this matrix has positive determinant.

Moreover, these matrix components may be interpreted as a $G$-invariant vector bundle isomorphism 
$$
\widetilde{\Phi} \, : \, \widetilde{X}\times \frg \xrightarrow{\quad \sim \quad} \widetilde{X}\times \FR^k
$$
over $\widetilde{X}$. Equivalently, this is simply a $G$-invariant map valued in $\frg^*\otimes \FR^k$
$$
\widetilde{\Phi}\, : \, \widetilde{X} \xrightarrow{\quad \, \quad} \frg^*\otimes \FR^k \, , 
$$
or yet equivalently, $G$-equivariant with respect to the \textit{trivial} coadjoint action of the abelian $G$ on $\frg^*$ tensored with the trivial action on $\FR^k$. Thus, via Eq. \eqref{VectorBundleValuedAreFormsTotalSpaceHorizontalForms}, this corresponds precisely to a section of the trivial associated bundle $(\frg^*\otimes \FR^k)_G \cong \bosonic{X}\times (\frg^*\otimes \FR^k)$ over the dimensionally reduced base $\bosonic{X}$, i.e., a simply a map 
$$
 \Phi : \, \bosonic{X} \xrightarrow{\quad \, \quad} \frg^*\otimes \FR^k \, .
$$
\end{proof}
\end{corollary}
In the case of a 1-dimensional fiber $G=U(1)$, the Lie algebra is one-dimensional and  $\frg^*\otimes \FR^1 \cong \FR^1$, so that the matter field $\Phi$ reduces to the prototypical case: the \textit{dilaton} field 
$$
\Phi\, :\, \bosonic{X}\xrightarrow{\quad \quad} \FR^+\hookrightarrow \FR \, , 
$$
which is positive-valued since $\mathrm{det}(\Phi)\equiv \Phi>0$. Of course, the full content of Cor. \ref{KK-FieldContentOfG-InvariantCoframe} recovers precisely the (0-modes of the) Kaluza--Klein reduction of a $G$-invariant metric $g$ via a corresponding partial gauge-fixing of its coframe gauge equivalency class. This is in line with the original literature (see e.g. \cite{ScSc79}\cite{Cr81}\cite{DNP86}), where however only the case of a \textit{trivial torus bundle} $\widetilde{X}=\bosonic{X}\times T^{ k}$ is considered, with no further justification on the usage of the naive Lie derivative and hence the explicit assumed form of the coframe. Our results consistently account both for the vanishing of the naive Lie derivative as a symmetry condition, along with the possible non-parallelizability of the reduced base spacetime $\bosonic{X}$ and the possible non-trivial topology of the $G$-bundle $\widetilde{X}\rightarrow X$ over it. 

\begin{remark}[\bf Reduction along timelike directions] The statements and proofs of Lem. \ref{InvarianceViaKosmannAndNaiveLieDerivativeOfCoframe}, Cor. \ref{InvarianceOfAssociatedFields} and Cor. \ref{KK-FieldContentOfG-InvariantCoframe} hold essentially verbatim in slight more generality for \textit{non-null} fundamental vector fields, i.e., including the case where the abelian $G$-fiber is $1$-dimensional and the fundamental vector field is \textit{timelike}. The only differences at this level are essentially cosmetic, by changing accordingly the signature of the corresponding reduced special orthogonal structure group to be strictly positive.
\end{remark}

\begin{remark}[\bf General non-abelian fibers]The generalization of Lem. \ref{InvarianceViaKosmannAndNaiveLieDerivativeOfCoframe} and Cor. \ref{InvarianceOfAssociatedFields} to the case of a general compact free, but \textit{non-abelian}, $G$-action on a manifold $\widetilde{X}$ supplied with an $\mathrm{SO}(1,d)$-structure and (Kosmann) $G$-invariant coframe is not straightfoward. In more detail, even though a decomposition of the corresponding metric à la Kaluza--Klein   
\begin{align}
g\, = \,\pi^*h+ g^V(\theta,\, \theta)
\end{align}
exists for any group $G$ (cf. proof of Lem. \ref{InvarianceViaKosmannAndNaiveLieDerivativeOfCoframe}), it remains unclear that the vertical component $g^V(\theta,\theta)$ has an orthonormal  coframe decomposition that is \textit{$G$-invariant} via the traditional Lie derivative. This is due to the fact that the connection $\theta$ is not $G$-invariant for a non-abelian group $G$, but is instead $G$-equivariant. Consequently, the GS-algorithm employed in the proof of Lem. \ref{InvarianceViaKosmannAndNaiveLieDerivativeOfCoframe} yields a \textit{$G$-equivariant} orthonormal coframe (with respect to the naive Lie derivative), rather than the $G$-invariant one required.
\end{remark}

\subsubsection*{Lifting the discussion to super-spacetimes}

For any rigorous considerations of dimensional reduction on super-spacetimes, and in particular for our forthcoming goal of dimensionally reducing the $S^4$-flux quantized \textit{super-space} formulation of 11D supergravity \cite{GSS24-SuGra} to a $\mathrm{cyc}(S^4)$-flux quantized super-space formulation of IIA 10D supergravity \cite{GS25b}, it is necessary to properly formulate a \textit{covariant notion} of a \textit{symmetric} supergravity \textit{super-spacetime} background and justify the usage of the naive Lie derivative as a special case. By the rheonomy property of super-space supergravity solutions, however, the demand that the corresponding (covariant) symmetry condition reduces to the vanishing of the Kosmann Lie derivatives upon restriction to the underlying bosonic spacetime already suggests the correct definition.

\begin{definition}[\bf Symmetric super-spacetime background]\label{SymmetricSuperspacetimeBackground}
A \textit{bosonic} vector field $\xi \in \CX^\even(X)$\footnote{Namely, $\CX^\even(X)$ here denotes the space of sections of the even subbundle  
$e^{-1} \, : \, \FR^{1,d}_{\mathrm{SO}(1,d)} \xrightarrow{\quad \sim \quad} T^{\even}X\hookrightarrow TX$ identified by the bosonic component of the super-coframe $(e^s, \psi^s)\, : \, TX\xrightarrow{\quad \sim \quad} \FR^{1,d \vert \mathbf{N}}_{\mathrm{SO}(1,d)}$.} is an (infinitesimal) symmetry of a
super-spacetime (\cite[Def. 2.74]{GSS24-SuGra})\footnote{The condition imposed on the (bosonic part of the super-) Torsion $\dd^{\omega^s} e = \big(\, \overline{\psi^s}\, \Gamma \psi^s\,\big)$ from \cite[Def. 2.74]{GSS24-SuGra} is not necessary for what follows.} 
$$
\big(X, \, (e^s,\psi^s,\om^s)\big)
$$
if the Kosmann Lie derivative along $\xi$ of the gravitational field vanishes
$$
(L_\xi^K e^s, \, L_\xi^K \psi^s, \, L^K_\xi \om^s) \, = \, 0 \, ,
$$
defined (formally) by the same formulas as in the purely bosonic manifold case \eqref{KosmannCovariantLieDerivativeFormulaAsCorrection}. 
\end{definition}
In more detail, this means that the Kosmann Lie derivative on a super-spacetime, along a bosonic vector field, is the covariant Lie derivative defined via the covariantization correction (Prop. \ref{EquivalentWaysOfDefiningACovariantLieDerivative}) given by\footnote{Note although the Levi--Civita connection on a super-spacetime is determined algebraically by the same condition, $\dd e^s + (\omega^{s})^{\mathrm{LC}}\wedge e^s = 0$, it does in general have non-trivial legs along the odd-coframe $(\omega^{s})^{\mathrm{LC}}= (\omega^{s})^{\mathrm{LC}}_a e^{s,a} + (\omega^{s})^{\mathrm{LC}}_\beta \psi^{s,\beta}$.} 
\begin{align*}
(B^K_\xi)^{a}{}_b\, :&= \, (\iota_\xi (\omega^{s})^{\mathrm{LC}} - \lambda_\xi^K)^a{}_b \\
 &= -\big( \eta^{ad}\, (L_\xi e^s)_{[bd]}\big) 
\end{align*} 
with $(\lambda^K_\xi)^{a}{}_b \, = \, \eta^{ad} \, (\dd^{\omega^{\mathrm{LC}}}\iota_\xi e^s)_{[bd]}\, .$ The reason for restricting to bosonic vector fields is that dimensional reduction occurs along \textit{bosonic} (abelian) $G$-fibers. The following result may be seen as a further justification for using the same formula for the Kosmann Lie derivative along such bosonic vector fields in the super-manifold setting, generalizing the corresponding property of the purely bosonic spacetimes from \eqref{KosmannCondition}.

\begin{lemma}[\bf Kosmann condition on super-spacetime]\label{KosmannConditionOnSuperSpacetime}
Consider the $(\even,\even)$ component of the super-metric corresponding to a super-coframe $(e^s,\psi^s)$ 
$$g^{\even} \, := \, \langle e^s,\, e^s\rangle \, \equiv  \eta_{ab} \, e^{s,a}\otimes e^{s,b} \quad \in \quad \Omega^1(X)\otimes_{\mathrm{Sym}} \Omega^1(X)\, . $$
The Kosmann Lie derivative of the bosonic part $e^s$ of a super-coframe $(e^s,\psi^s)$, along a bosonic vector field $\xi \in \CX^\even(X)$ satisfies
\begin{align}\label{SuperspaceKosmannCondition}
L_\xi g^\even \, = 2 \langle L_\xi^K e^s \, ,  e^s \rangle \equiv \, 2\,  \eta_{ab}  \, (L_{\xi}^K e)^{s,a} \, \otimes \, e^{s,b} \, = \,  \,2\,   L_{\xi}^K e^s_b\, \otimes \, e^{s,b}\, . 
\end{align}
In fact, this condition completely determines the Kosmann lift of bosonic vector fields to the corresponding $\mathrm{SO}(1,d)$ principal bundle and further its $\mathrm{Spin}(1,d)$-refinement.
\end{lemma}
\begin{proof}
Given our restriction to bosonic vector fields and the $(\even,\even)$ part of the super-metric, this follows formally as the proof of Prop. \ref{TheKosmannLiftExists} (and Rem. \ref{LiftingToSpinCovers}) apart from one point that requires justification. Namely, a $\mathrm{SO}(1,d)$ covariant Lie derivative of the coframe along an arbitrary vector field is given by
$$
\widetilde{L}_\xi \, = \, L_\xi e^s + B_\xi \cdot e^s \, , 
$$
and so it is (a priori) plausible that the traditional Lie derivative term $L_\xi e^s$ might have non-trivial components along the odd coframe $\psi^s$, and hence the Lie derivative of the corresponding even metric $L_\xi g^\even$ too. However, this is impossible for \textit{bosonic} vector fields $\xi \in \CX^\even(X)$, since this distribution is involutive on any supermanifold $X$ (by the even-parity of the Lie bracket). 

The details of this latter argument are standard but we include these for completeness. Let $(\hat{e}^s{}_a)_{a=0,\cdots, d}$ be the dual frame to the bosonic coframe $e^s=(e^{s,a})_{a=0,\cdots, d}$, so that any bosonic vector field expands as $
\xi \, = \, \xi^a \hat{e}^s{}_a.
$ It follows that the (locally defined) traditional Lie derivative of the frame is given by
\begin{align*}
L_\xi \hat{e}^s{}_b\, &= \, [\xi^a \hat{e}^s{}_a,\hat{e}^s{}_b] \, = \, \xi^a [\hat{e}^s{}_a, \hat{e}^s{}_b] - \hat{e}^s{}_b(\xi^a)\,  \hat{e}^s{}_a \\
&=\, \xi^a C^c{}_{ab}\,  \hat{e}^s{}_c - \hat{e}^s{}_b(\xi^a)\,  \hat{e}^s{}_a \quad \qquad \in \quad \qquad \CX^\even(X)\, ,
\end{align*}
for some structure functions $\{C^c{}_{ab}\}_{a,b,c=0,\cdots,d}\subset C^\infty(X)$, since the Lie bracket of two even vector fields is an even vector field. Acting with the Lie derivative on the defining duality condition $e^{s,d}(\hat{e}^s{}_b)= \delta^{d}{}_b$ yields
$$
L_\xi e^{s,d} \, = \, \big(-\xi^a \, C^{d}{}_{ab} + \hat{e}^s{}_b(\xi^d)\big) e^{s,b} \, ,
$$
which manifestly has only bosonic coframe components. 
\end{proof}

\begin{remark}[\bf Lifting odd vector fields]
We stress that the above Kosmann lift and corresponding covariant Lie derivatives are only \textit{partially defined}, namely only for \textit{bosonic} vector fields on the super-spacetime. Naturally, given that Lem. \ref{KosmannConditionOnSuperSpacetime} fully determines the lift by a condition on the \textit{even} component of the super-metric
\begin{align*}
g \, &= \, g^{\even} + g^\odd \, := \, \langle e^s,\, e^s\rangle + \langle \psi^s,\, \psi^s \rangle 
\\
&\equiv \,   \eta_{ab} \, e^{s,a}\otimes e^{s,b} + \eta_{\alpha \beta} \, \psi^{s,\alpha} \otimes \psi^{s, \beta}\, , 
\end{align*}
where $\eta_{\alpha \beta}= -\eta_{\beta \alpha}$ are the components of the (flat) ``symplectic metric'', we expect that the lift on odd vector fields $\CX^\odd(X)$ may be determined by the further (separate) demand that $$
L_{\xi^\odd} g^\odd \, = \, 2\,  \eta_{\alpha \beta} \, L_{\xi^\odd}^K \psi^{s,\alpha} \otimes \psi^{s,\beta}\, .
$$ 
Although this is an interesting question worth further study, we do not pursue it here as it is outside the scope of our motivation of dimensional reduction, which takes place along \textit{bosonic fibers}.
\end{remark}

With Def. \ref{SymmetricSuperspacetimeBackground} and Lem. \ref{KosmannConditionOnSuperSpacetime} established on super-spacetimes, the contents of Lem. \ref{InvarianceViaKosmannAndNaiveLieDerivativeOfCoframe}, Cor. \ref{InvarianceOfAssociatedFields} and Cor. \ref{KK-FieldContentOfG-InvariantCoframe} from symmetric bosonic spacetimes
follow formally in the same manner for symmetric super-spacetimes. We summarize these in the following.

\begin{corollary}[\bf Invariance via the naive Lie derivative on super-spacetime]\label{InvarianceViaTheNaiveLieDerivativeOnSuper-spacetime}
Let $\widetilde{X}$ be a super-principal $G$-bundle $\pi: \widetilde{X}\longrightarrow X$, with $G$ \textit{bosonic} and \textit{abelian}, supplied with a $G$-symmetric super-spacetime structure $(e^s, \psi^s,\omega^s)$, i.e.,   
$$
(L_{A^\#}^K e^s, \, L_{A^\#}^K \psi^s, \, L^K_{A^\#} \om^s) \, = \, 0 \, ,
$$
for all fundamental vector fields $A^\# \in \Gamma_{\widetilde{X}}(V\widetilde{X})$ generated by the $G$-action, corresponding to  elements of the Lie algebra $A \in \frg = \mathrm{Lie}(G)$. Assuming the fundamental vector fields are spacelike for the corresponding $\even$ metric 
$$
g^\even \, = \, \langle  e^s , e^s \rangle \, = \, \eta_{ab}\, e^{s,a}\otimes e^{s,b} \, ,
$$
then:
\begin{itemize}
\item 
There always exists a gauge equivalent super-spacetime field configuration 
\begin{align*}
(e^s, \psi^s,\omega^s) \longmapsto (\hat{e}^s, \hat{\psi}^s,\hat{\omega}^s) \,= \, r \cdot (e^s, \psi^s,\omega^s)
\end{align*}
for some $\mathrm{Spin}(1,d)$-gauge transformation $r: \widetilde{X}\rightarrow \mathrm{Ad}(\widetilde{P}) $
with respect to the corresponding $\mathrm{Spin}(1,d)$-structure $\widetilde{P}\rightarrow \widetilde{X}$,  for which the vanishing of the Kosmann Lie derivatives are equivalent to the vanishing of the naive Lie derivatives 
$$
(L_{A^\#} \hat{e}^s, \, L_{A^\#} \hat{\psi}^s, \, L_{A^\#} \hat{\om}^s) \, = \, 0 \, .
$$

\item For $\mathrm{dim}(G)=k$, there exists a gauge sub-equivalence super-spacetime field configurations with the same property, given by the orbit of $(\hat{e}^s,\hat{\psi}^s, \hat{\omega}^{s})$ under $G$-invariant local $\mathrm{Spin}(1,d-k)\times \mathrm{Spin}(k)$-transformations with respect to an accordingly reduced $\mathrm{Spin}(1,d-k)\times\mathrm{Spin}(k)$-structure $\hat{P}\rightarrow \widetilde{X}$.

\item The partially gauge fixed $G$-invariant bosonic coframe $\hat{e}^s$ encodes precisely the following fields on the dimensionally reduced super-spacetime $X = \widetilde{X}/G$:
\begin{itemize}

\item[\bf (i)] a bosonic coframe $\hat{e}^s_h: TX \longrightarrow \FR^{1,d-k}_{\mathrm{SO}(1,d-k)}$, 
\item[\bf (ii)] A family of locally defined (bosonic) $G$-gauge fields $
\big\{\theta^s_i \quad  \in \quad  \Omega^1( U_i; \, \mathfrak{g}) \big\}_{i\in I}\, ,
$
on a trivializing (super)-cover  $\{U_i \hookrightarrow X\}_{i\in I}$, 
\item[\bf (iii)] A globally defined (bosonic) scalar matter field $\Phi^s: X \longrightarrow \frg^* \otimes \FR^k$.
\end{itemize}
\end{itemize}
\end{corollary}
\begin{proof}
The proof follows as those of Lem. \ref{InvarianceViaKosmannAndNaiveLieDerivativeOfCoframe}, Cor. \ref{InvarianceOfAssociatedFields} and Cor. \ref{KK-FieldContentOfG-InvariantCoframe} by parsing them in the category of super manifolds with minimal (cosmetic) modifications, such as restricting to the even tangent bundle $T^\even X\hookrightarrow TX$ and the corresponding even metric $g^\even$, while performing an appropriately modified version of the Gram--Schmidt process on supermanifolds \cite[Sec. 2.8]{DeW84}. Thus we shall refrain from repeating the full proofs here.
\end{proof}

\section{Summary \& Outlook}
We have provided a concise, rigorous and complete mathematical description of covariant Lie derivatives of $G$-gauge fields and associated matter fields, which previously existed only in fragmentary form. This was done from both the perspective of lifts to the corresponding principal $G$-bundles (Def. \ref{CovariantLieDerivative}) and also from that of the base spacetime via covariantization correction formulas (Lem. \ref{CovariantLieDerivativeOnThebAseViaAConnection}), while exhibiting the precise equivalence between the two approaches (Prop. \ref{EquivalentWaysOfDefiningACovariantLieDerivative}).

We have detailed two occurences of covariant Lie derivatives in the context of first-order (super-)gravity: {\bf (i)} the natural covariant Lie derivative (Ex. \ref{LiftingViaChosenConnection}), employed to fully explain the relation of on-shell (super-)diffeomorphism symmetry to local translational (super-)symmetry (Sec. \ref{SuperDiffeomorphismSymmetrySec}), and {\bf (ii)} the Kosmann Lie derivative (Def. \ref{KosmannLift}, Prop. \ref{TheKosmannLiftExists}), appropriate for acting on coframes and associated spinorial fields towards the description of background isometries via its vanishing condition (Cor. \ref{TheVanishingofKosmannLieAndTheKillingEquation}).

Finally, we have employed the \textit{covariant} vanishing symmetry condition of the latter to exhibit a special gauge where it coincides with the vanishing of the traditional \textit{non-covariant} Lie derivative (Lem. \ref{InvarianceViaKosmannAndNaiveLieDerivativeOfCoframe}, Cor. \ref{InvarianceOfAssociatedFields}), for the case of abelian $G$-symmetries on a principal $G$-bundle spacetime. As a further by-product, we have detailed the dimensionally reduced field content encoded in this partially fixed-gauge total spacetime coframe, consistently applicable to any principal $G$-bundle topology on the total spacetime and any topology on the dimensional reduced spacetime (Cor. \ref{KK-FieldContentOfG-InvariantCoframe}).

The latter result is crucial for guaranteeing the non-triviality of our forthcoming application in the dimensional reduction of $S^4$-flux quantized 11D super-space supergravity \cite{GSS24-SuGra} to $\mathrm{cyc}(S^4)$-flux quantized IIA 10D super-space supergravity \cite{GS25b}, and further potential toroidal reductions thereof, since flux quantization yields new non-trivial information exclusively on spacetimes of non-trivial topology.

\vspace{1cm} 
\noindent {\bf Acknowledgements} \,  
The author is thankful to Urs Schreiber for comments on an earlier draft of this text.

\vspace{1cm}
\noindent {\bf Data availability} \, All data generated or analyzed during this study are contained in this document.

\vspace{1cm}
\noindent {\bf Conflict of interest} \, The author states that there is no conflict of
interest.

\end{document}